\documentclass[11pt]{article}
\def\diag{\operatorname{diag}}

\def\RR{\mathbb{R}}

\def\max{\mathrm{max}}
\def\min{\mathrm{min}}

\def\stab{\operatorname{stab}}

\def\geom{\operatorname{geom}}

\def\S{\mathcal{S}}

\def\C{\mathcal{C}}

\def\OOplus={{ {{{\mathcal{O}}}}_+}}
\def\OObar={\overline {\mathcal{O}_+}}
\def\RR{\mathbb{R}}

\def\IR{\mathbb{R}}

\def\P{\mathcal{P}}

\def\sgn{\operatorname{sgn}}

\newcommand{\od}{\stackrel{\mbox {\tiny {def}}}{=}}

\newcommand{\norm}[1]{\Vert #1\Vert}

\newcommand{\supp}{\operatorname{supp}}

\def\C{\mathcal{C}}

\def\CM{\mathrm{cm}}

\usepackage{color}
\definecolor{gold}{rgb}{0.85,.66,0}
\definecolor{cherry}{rgb}{0.9,.1,.2}
\definecolor{burgundy}{rgb}{0.8,.2,.2}
\definecolor{orangered}{rgb}{0.85,.3,0}
\definecolor{orange}{rgb}{0.85,.4,0}
\definecolor{olive}{rgb}{.45,.4,0}
\definecolor{lime}{rgb}{.6,.9,0}

\definecolor{green}{rgb}{.2,.7,0}
\definecolor{grey}{rgb}{.4,.4,.2}
\definecolor{brown}{rgb}{.4,.2,.1}
\definecolor{blue}{rgb}{0,.0, .81}

\def\ssec#1{\noindent {\bf #1.}}
\def\q#1#2{\noindent {\it Question #1: #2}}
\usepackage{graphicx, wrapfig}
\usepackage{amssymb,amsfonts,amsmath,amsthm,url}
\usepackage{simplemargins}
\usepackage{ulem}
\newtheorem{corollary}{Corollary}

\newtheorem{proposition}{Proposition}
\newtheorem{theorem}{Theorem}
\newtheorem{lemma}{Lemma}

\newenvironment{definition}[1][Definition.]{\begin{trivlist}
\item[\hskip \labelsep {\bfseries #1}]}{\end{trivlist}}
\begin{document}

\title{Encoding binary neural codes in networks \\ of threshold-linear neurons}
%\titlerunning{Encoding binary patterns in neural networks}
\author{Carina Curto$^1$ $\cdot$ Anda Degeratu$^2$ $\cdot$ Vladimir Itskov$^1$}

\date{May 16, 2013}
%\authorrunning{C. Curto, A. Degeratu, V. Itskov}
%\institute{C. Curto \at
%              University of Nebraska-Lincoln, Lincoln, NE\\
%              \email{ccurto2@math.unl.edu}           %  \\
%%             \emph{Present address:} of F. Author  %  if needed
%           \and
%           A. Degeratu \at
%           Albert-Ludwig-Universit\"at, Freiburg, Germany\\
%           \email{anda.degeratu@math.uni-freiburg.de}
%           \and
%              V. Itskov \at
%              University of Nebraska-Lincoln, Lincoln, NE\\
%              \email{vladimir.itskov@math.unl.edu}         
%}
%\orangered{Version: November 30, 2012}}

\maketitle
$^1$ Department of Mathematics, University of Nebraska-Lincoln\\
\indent $^2$ Department of Mathematics, Albert-Ludwig-Universit\"at, Freiburg, Germany\\
\indent\; ccurto2@math.unl.edu, anda.degeratu@math.uni-freiburg.de, vladimir.itskov@math.unl.edu\\
\vspace{.25in}

\begin{abstract} 
Networks of neurons in the brain encode preferred patterns of neural activity via their synaptic connections.  
Despite receiving considerable attention, the precise relationship between network connectivity and encoded patterns is still poorly understood.  
Here we consider this problem for networks of threshold-linear neurons whose computational function
is to learn and store a set of binary patterns (e.g., a {\it neural code}) as ``permitted sets'' of the network.  We introduce a simple Encoding Rule that selectively turns ``on''  synapses between neurons that co-appear in one or more patterns.  The rule uses synapses that are {\it binary}, in the sense of having only two states (``on'' or ``off''), but also {\it heterogeneous}, with weights drawn from an underlying  synaptic strength matrix $S$.  Our main results precisely describe the stored patterns that result from the Encoding Rule -- including unintended ``spurious'' states -- and give an explicit characterization of the dependence on $S$.  
In particular, we find that binary patterns are successfully stored in these networks when the excitatory connections
between neurons are {\it geometrically balanced} -- i.e., they satisfy a set of geometric constraints.
Furthermore, we find that certain types of neural codes are {\it natural} in the context of these networks, meaning that
the full code can be accurately learned from a highly undersampled set of patterns.  
Interestingly, many commonly observed neural codes in cortical and hippocampal areas are natural in this sense.
As an application, we construct networks that encode hippocampal place field codes nearly exactly,
following presentation of only a small fraction of patterns.
To obtain our results, we prove new theorems
using classical ideas from convex and distance geometry, such as Cayley-Menger determinants, revealing a novel connection between these areas of mathematics and coding properties of neural networks. 
\end{abstract}

\begin{small}
\tableofcontents
\end{small}
\bigskip

\section{Introduction}

Recurrent networks in cortex and hippocampus exhibit highly constrained patterns of neural activity, even in the absence of sensory inputs \cite{kenet:2003,yuste:2005,luczak:2009,fiser:2011}.  These patterns are strikingly similar in both stimulus-evoked and spontaneous activity \cite{kenet:2003,luczak:2009}, 
suggesting that cortical networks store neural codes consisting of a relatively small number of allowed activity patterns \cite{yuste:2005,fiser:2011}.  What is the relationship between the stored patterns of a network and its underlying connectivity? 
 More specifically, given a prescribed set of binary patterns (e.g., a binary neural code), how can one arrange the connectivity of a network such that precisely those patterns are encoded as fixed point attractors of the dynamics, while minimizing the emergence of unwanted ``spurious'' states?  This problem, which  we refer to as the {\it Network Encoding (NE) Problem}, dates back at least to 1982 and has been most commonly studied in the context of the Hopfield model \cite{Hopfield:1982,Hopfield:1984,Amitbook,HertzKroghPalmer}.   A major challenge in this line of work has been to characterize the spurious states \cite{amit_spin-glass_1985, AGS87, Amit1989, HertzKroghPalmer, RoudiTreves2003}.

In this article we take a new look at the NE problem for networks of threshold-linear neurons whose computational function is to learn and store binary neural codes.  Following \cite{Seung:2002,Seung:2003}, we regard stored patterns of a threshold-linear network as ``permitted sets'' (a.k.a. ``stable sets'' \cite{CDI12}), corresponding to subsets of neurons that may be co-active at stable fixed points of the dynamics in the presence of one or more external inputs.  Although our main results do not make any special assumptions about the prescribed sets of patterns to be stored, many commonly observed neural codes are sparse and have a rich internal structure, with correlated patterns reflecting similarities between represented stimuli.   Our perspective thus differs somewhat from the traditional Hopfield model \cite{Hopfield:1982, Hopfield:1984}, where binary patterns are typically assumed to be uncorrelated and dense \cite{Amitbook,HertzKroghPalmer}. 

To tackle the NE problem we introduce a simple learning rule, called the Encoding Rule, that constructs a network $W$ from a set of prescribed binary patterns $\C$.
The rule selectively turns ``on'' connections between neurons that co-appear in one or more of the presented patterns, and uses synapses that are {\it binary} (in the sense of having only two states -- ``on'' or ``off''), but also {\it heterogeneous}, with weights drawn from an underlying  synaptic strength matrix $S$.  
Our main result, Theorem~\ref{thm:main-result}, precisely characterizes the full set of permitted sets 
$\P(W)$ for any network constructed using the Encoding Rule, and shows explicitly the dependence
on $S$.  In particular, we find that binary patterns can be successfully stored in these networks if and only if the strengths of excitatory connections
among co-active neurons in a pattern are {\it geometrically balanced} -- i.e., they satisfy a set of geometric constraints.  
Theorem~\ref{thm:main-result2} shows that any set of binary patterns that can be exactly encoded 
as $\C = \P(W)$ for symmetric $W$ can in fact be exactly encoded using our Encoding Rule.
Furthermore, when a set of binary patterns $\C$ is {\it not} encoded exactly, we are able to completely describe the spurious states, and find that they correspond to cliques in the ``co-firing'' graph $G(\C)$.  

An important consequence of these findings is that certain neural codes are {\it natural} in the context of symmetric threshold-linear networks;
i.e., the structure of the code closely matches the structure of emerging spurious states via the Encoding Rule, 
allowing the full code to be accurately learned from a highly undersampled set of patterns.  Interestingly, using Helly's theorem \cite{Barvinok2002} 
we can show that many commonly observed neural codes in cortical and hippocampal areas are natural in this sense.
As an application, we construct networks that encode hippocampal place field codes nearly exactly,
following presentation of only a small and randomly-sampled fraction of patterns in the code.

The organization of this paper is as follows.  In Section~\ref{sec:background} we introduce some necessary background on binary neural codes, threshold-linear networks, and permitted sets.  In Section~\ref{sec:results} we introduce the Encoding Rule and present our Results.  The proofs of our main results are given in Section~\ref{sec:proofs} and use ideas from classical distance and convex geometry, such as Cayley-Menger determinants \cite{Blumenthal}, 
establishing a novel connection between these areas of mathematics and neural network theory.  
Finally, Sections~\ref{sec:discussion} and~\ref{sec:appendices} contain the Discussion and Appendices, respectively.

\section{Background} \label{sec:background}

 \subsection{Binary neural codes} \label{sec:binary}

A binary pattern on $n$ neurons is simply a string of $0$s and $1s$, with a $1$ for each active neuron and a $0$ denoting silence; 
equivalently, it is a subset of (active) neurons 
$$\sigma \subset \{1,\ldots,n\} \od [n].$$  
A {\it binary neural code} (a.k.a. a combinatorial neural code \cite{neuro-coding,OsborneBialek08}) is a collection of binary patterns $\C \subset 2^{[n]}$, where $2^{[n]}$ denotes the set of all subsets of $[n]$.  

Experimentally observed neural activity in cortical and hippocampal areas suggests that neural codes are {\it sparse} \cite{Hromadka:2008,Barth:2012}, meaning
that relatively few neurons are co-active in response to any given stimulus. 
Correspondingly, we say that a binary neural code $\C \subset 2^{[n]}$ is {\it $k$-sparse}, for $k < n$, if 
all patterns  $\sigma \in \C$ satisfy $|\sigma| \leq k$.  Note that in order for a code $\C$ to have good error-correcting capability, the total number of codewords $|\C|$ must be considerably smaller than $2^n$ \cite{MacWilliamsSloane83,HuffmanPless03,neuro-coding}, a fact that may account for the limited repertoire of observed neural activity.  

Important examples of binary neural codes are classical population codes, such as {\it receptive field codes} (RF codes) \cite{neuro-coding}.
A simple yet paradigmatic example is the hippocampal {\it place field code} (PF code), where single neuron activity is characterized by place fields \cite{O:1976,ON:1978}. 
We will consider general RF codes in Section~\ref{sec:natural}, and specialize to sparse PF codes in Section~\ref{sec:PFcodes}.

\subsection{Threshold-linear networks} 
A {\it threshold-linear network} \cite{Seung:2003, CDI12} is a firing rate model for a recurrent network \cite{DayanAbbott,ErmentroutTerman} 
where the neurons all have threshold nonlinearity, $\phi(z) = [z]_+ = \max\{z,0\}$.  
The dynamics are given by
\begin{equation*} \frac{d x_i}{dt}=-\frac1{\tau_i}x_i+\phi\left(\sum_{j=1}^n  W_{ij}x_j+e_i-\theta_i  \right), \; \;\; i = 1,...,n,
\end{equation*}
where $n$ is the number of neurons,  $x_i(t)$ is the firing rate of the $i$th neuron at time $t$, $e_i$ is the external input to the $i$th neuron, and $\theta_i>0$ is its threshold.  
The matrix entry $W_{ij}$ denotes the effective strength of the connection from the $j$th to the $i$th neuron, and the timescale $\tau_i > 0$ gives the rate at which a neuron's activity decays to zero in the absence of any inputs (see Figure 1).
 
\begin{wrapfigure}{r}{.4\linewidth} 

%\vspace{-.15in}

\begin{center}
\includegraphics[width=2.25in]{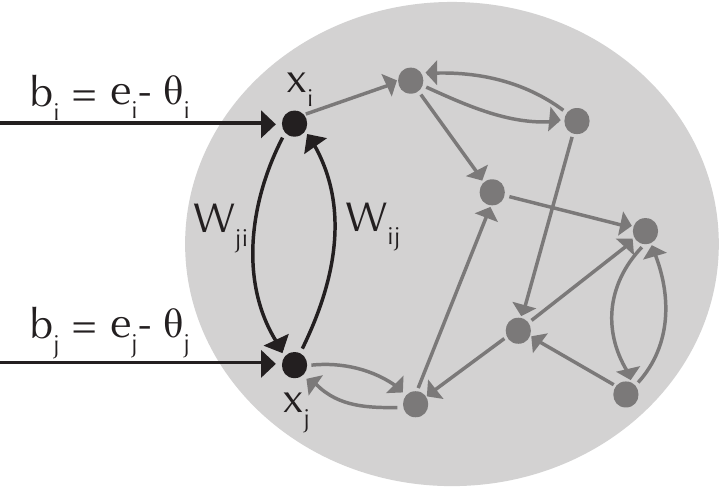}
\end{center}

%\vspace{-.15in}

\caption{\small A recurrent network receiving an input vector $b = (b_1,\ldots,b_n)$.  The firing rate of each neuron
is given by $x_i = x_i(t),$ and evolves in time according to equation~\eqref{eq:dynamics}.  The strengths of recurrent connection are captured by the matrix $W$.}
\end{wrapfigure}

Although sigmoids more closely match experimentally measured input-output curves for neurons, the above threshold nonlinearity
is often a good approximation when neurons are far from saturation \cite{DayanAbbott,Shriki:2003}.  Assuming that encoded patterns of 
a network are in fact realized by neurons that are firing far from saturation, it is reasonable to approximate 
them as stable fixed points of the threshold-linear dynamics.

These dynamics can be expressed more compactly as
\begin{equation} \label{eq:dynamics}
 \dot x = -Dx+\left[W x + b  \right]_+\!,
\end{equation}
where $D\od \operatorname{diag}(1/{\tau_1},...,1/{\tau_n})$ is the diagonal matrix of inverse time constants, $W$ is the synaptic connectivity matrix,
$b = (b_1,...,b_n) \in \RR^n$ with $b_i = e_i-\theta_i$, and $[ \cdot ]_+$ is applied elementwise.  Note that, unlike in the Hopfield model, the ``input'' to the network comes in the form of a constant (in time) external drive $b$, rather than an initial condition $x(0)$.  We think of equation~\eqref{eq:dynamics} as describing the fast-timescale dynamics
of the network, and $b$ as representing the effect of an external stimulus.  So long as $b$ changes slowly as compared to the fast network dynamics, the neural
responses to individual stimuli are captured by the steady states of~\eqref{eq:dynamics} in the presence of a constant input vector $b$.

 In the Encoding Rule (Section~\ref{sec:encoding-rule}), we assume homogeneous
 timescales and use $D = I$ (the identity matrix).  Nevertheless, all results apply equally well to heterogeneous timescales -- i.e., for any diagonal $D$ having strictly positive diagonal.  We also assume that $-D+W$ has strictly negative diagonal, so that the activity of an individual neuron always decays to zero in the absence of external or recurrent inputs.  Although we consider responses to the full range of inputs $b \in \RR^n$, the possible steady states of~\eqref{eq:dynamics} are sharply constrained by the connectivity matrix $W$.  Assuming fixed $D$, we refer to a particular threshold-linear network simply as $W$.

\begin{table}
\begin{center}
\begin{tabular}{l|l}
notation & meaning\\
\hline \hline
$[n]$ & $\{1,...,n\}, \;\; n = \text{ \# of neurons}$\\
$2^{[n]}$ & the set of all subsets of $[n]$\\
$\sigma \subset [n]$ & a subset of neurons; a binary pattern; a codeword; a permitted set\\
$|\sigma|$ & \# of elements (neurons) in the set $\sigma$\\
$\C \subset 2^{[n]}$ & a prescribed set of binary patterns -- e.g., a binary neural code\\
$G(\C)$ & the co-firing graph of $\C$; $(ij) \in G(\C) \Leftrightarrow \{i,j\} \subset \sigma$ for some $\sigma \in \C$\\
$X(G), X(G(\C))$ & the clique complex of the graph $G$ or $G(\C)$, respectively\\
$\supp(x)$ & $\{i \in [n] \mid x_i>0 \}$, for $x \in \RR^n_{\geq 0}$ a nonnegative vector\\
$W$ & an $n \times n$ connectivity matrix; the network with dynamics~\eqref{eq:dynamics}\\
$D$ & fixed diagonal matrix of inverse time constants\\
$\P(W)$ & $\{\sigma \subset [n] \mid \sigma \text{ a permitted set of } W\}$; set of all permitted sets of $W$\\
$A$ & an $n \times n$ matrix\\
%$A$, Hebbian matrix & $A_{ij}=A_{ji} \geq 0$ and $A_{ii} = 0$\\ &\;\;\; for all $i,j \in [n]$\\
$A_\sigma, \text{ for } \sigma \subset [n]$ & the principal submatrix of  $A$ with index set $\sigma$\\
$\stab(A)$ & $\{\sigma \subset [n] \mid A_\sigma \text{ is a stable matrix}\}$\\ 
$\CM(A)$ & Cayley-Menger determinant of $A$\\
$1 \in \RR^n$ & the column vector with all entries equal to $1$\\
$-11^T$ & $n \times n$ rank 1 matrix with all entries equal to $-1$
%$\RR_{\times}^n$ & $\{v \in \RR^n \mid v_i \neq 0 \;\text{for all}\; i \in [n]\}$\\
%$A^v$, for $v \in \RR^n$ & $\diag(v)A\diag(v)$\\ 
%$-vv^T$, for $v \in \RR^n$ & rank 1 matrix with entries $-v_iv_j$\\
\end{tabular}
\caption{Frequently used notation.}
\end{center}
\end{table}

\subsection{Permitted sets of threshold-linear networks}\label{sec:permitted}

We consider threshold-linear networks whose computational function is to encode a set of binary patterns.  
These patterns are stored as ``permitted sets'' of the network.
The theory of permitted (and forbidden) sets was first introduced in \cite{Seung:2002, Seung:2003}, and many interesting results were obtained in the case of symmetric threshold-linear networks.  Here we review some definitions and results that apply more generally, though later we will also restrict ourselves to the symmetric case.

Informally, a {\it permitted set} of a recurrent network is a binary pattern $\sigma \subset [n]$ that can be {\it activated}. 
This means there exists an external input to the network such that the neural activity $x(t) = (x_1(t),\ldots,x_n(t))$ converges to a steady state $x^* \in \RR^n_{\geq 0}$ (i.e., $x^*$ is a stable fixed point with all firing rates nonnegative) having support $\sigma$:
$$\sigma = \supp(x^*) \od \{i \in [n] \mid x_i^* > 0\}.$$

\begin{definition}
A {\it permitted set} of the network~\eqref{eq:dynamics} is a subset of neurons $\sigma \subset [n]$ with the property that for at least one external input $b \in \RR^n$, there exists an asymptotically stable fixed point $x^{*} \in \RR^n_{\geq 0}$ such that $\sigma = \supp(x^*)$ \cite{Seung:2003}.  
For a given choice of network dynamics, the connectivity matrix $W$ determines the set  of all permitted sets of the network, denoted $\P(W)$.
\end{definition}

For threshold-linear networks of the form~\eqref{eq:dynamics}, 
it has been previously shown that permitted sets of $W$ correspond to stable principal submatrices of $-D+W$ \cite{Seung:2003,CDI12}.  Recall that a {\it stable} matrix is one whose eigenvalues all have strictly negative real part.  For any $n \times n$ matrix $A$, the notation $A_\sigma$ denotes the {\em principal submatrix} obtained by restricting to the index set $\sigma$; if $\sigma = \{s_1,...,s_k\}$, then $A_\sigma$ is the $k  \times k$ matrix with $(A_\sigma)_{ij} = A_{s_i s_j}$.  
We denote the set of all stable principal submatrices of $A$ as
$$\stab(A) \od \{\sigma \subset [n] \mid A_\sigma \text{ is a stable matrix}\}.$$
With this notation we can now restate our prior result, which generalizes an earlier result of \cite{Seung:2003} to non-symmetric networks.

\begin{theorem}[{\cite[Theorem 1.2]{CDI12}}\footnote{Note that in \cite[Theorem 1.2]{CDI12}, permitted sets were called ``stable sets.''  See also \cite{Seung:2003} for an earlier proof specific to the symmetric case.}]\label{thm:thm1}\label{thm:stable-set}
Let $W$ be a threshold-linear network on $n$ neurons (not necessarily symmetric) with dynamics given by equation~\eqref{eq:dynamics}, and let $\P(W)$ be the set of all permitted sets of $W$.  Then
$$\P(W) = \stab(-D+W).$$
\end{theorem}

Theorem~\ref{thm:thm1} implies that a binary neural code $\C$ can be exactly encoded as the set of permitted sets in a threshold-linear network if and only if there exists a pair of $n \times n$ matrices $(D,W)$ such that 
$\C = \stab(-D+W)$.
From this observation, it is not difficult to see that not all codes are realizable by threshold-linear networks.
This follows from a simple lemma.

\begin{lemma}\label{lemma:2x2}
Let $A$ be an $n \times n$ real-valued matrix (not necessarily symmetric) with strictly negative diagonal and $n \geq 2$.  If $A$ is stable, then there exists a $2 \times 2$ principal submatrix of $A$ that is also stable.
\end{lemma}

\begin{proof}
We use the formula for the characteristic polynomial in terms of sums of principal minors:
\begin{eqnarray*}
p_A(X) &=& (-1)^n X^n + (-1)^{n-1}m_1(A)X^{n-1} + (-1)^{n-2}m_2(A)X^{n-2} + .... + m_n(A),
\end{eqnarray*}
where $m_k(A)$ is the sum of the $k \times k$ principal minors of $A$.  Writing the characteristic polynomial in terms of symmetric polynomials in the eigenvalues $\lambda_1,\lambda_2,...,\lambda_n$, and assuming $A$ is stable, we have
$m_2(A) = \sum_{i<j} \lambda_i\lambda_j > 0.$
This implies that at least one $2 \times 2$ principal minor is positive.  Since the corresponding $2 \times 2$ principal submatrix has negative trace, it must be stable. 
\end{proof}

\noindent Combining Lemma~\ref{lemma:2x2} with Theorem~\ref{thm:stable-set} then gives:

\begin{corollary}\label{cor:2-restrict}
Let $\C \subset 2^{[n]}$.  If there exists a pattern $\sigma \in \C$ such that no order 2 subset of $\sigma$ belongs to $\C$, then it is not realizable as $\C = \P(W)$ for any threshold-linear network $W$.
\end{corollary}
  
Here we will not pay attention to the relationship between the input to the network $b$ and the corresponding permitted sets that may be activated, as it is beyond the scope of this paper.   In prior work, however, we were able to understand with significant detail the relationship between a given $b$ and the set of resulting fixed points of the dynamics \cite[Proposition 2.1]{CDI12}.  For completeness, we summarize these findings in Appendix D.
 
 \subsection*{Permitted sets of  symmetric threshold-linear networks}
In the remainder of this work, we will restrict attention to the case of symmetric networks.  With this assumption, we can immediately say more about the structure of permitted sets $\P(W)$.  Namely, if $W$ is symmetric then the permitted sets $\P(W)$ have the combinatorial structure of a simplicial complex.

\begin{definition} An (abstract) {\it simplicial complex} $\Delta \subset 2^{[n]}$ is a set of subsets of $[n] = \{1,\ldots,n\}$ such that the following two properties hold:
(1) for each $i \in [n], \{i\} \in \Delta$, and
(2) if $\sigma \in \Delta$ and $\tau \subset \sigma$, then $\tau \in \Delta$.
\end{definition}

\begin{lemma}\label{lemma:simplicial}
If $W$ is a symmetric threshold-linear network, then $\P(W)$ is a simplicial complex.
\end{lemma}

\noindent In other words, if $W$ is symmetric then every subset of a permitted set is permitted, and every superset of a set that is not permitted is also not permitted.   This was first observed in \cite{Seung:2003}, using an earlier version of Theorem~\ref{thm:stable-set} for symmetric $W$.  
It follows from the fact that $\P(W) = \stab(-D+W)$, by Theorem~\ref{thm:stable-set}, and $\stab(A)$ is a simplicial complex for any symmetric $n \times n$ matrix $A$ having strictly negative diagonal (see Corollary~\ref{cor:stab-simplicial} in Appendix A).  The proof of this fact is a straightforward application of Cauchy's interlacing theorem (Appendix A), which applies only to symmetric matrices.

We are not currently aware of any simplicial complex $\Delta$ that is not realizable as $\Delta = \P(W)$ for a symmetric threshold-linear network, although we believe such examples are likely to exist.
 
\section{Results} \label{sec:results}
 
Theorem~\ref{thm:stable-set} allows one to find all permitted sets $\P(W)$ of a given network $W$.  Our primary
interest, however, is in the inverse problem.  

\begin{quote}
\noindent{\it \textbf{NE problem}: Given a set of binary patterns $\C \subset 2^{[n]}$, how can one construct a network $W$ such that 
$\C \subseteq \P(W),$ while minimizing the emergence of unwanted spurious states?}
\end{quote}

\noindent Note that {\it spurious states} are elements of $\P(W)$ that were 
{\it not} in the prescribed set of patterns to be stored; these are precisely the elements of $\P(W) \setminus \C$.
If $\C \subset \P(W)$, so that all patterns in $\C$ are stored as permitted sets of $W$ but $\P(W)$ may contain additional spurious states, 
then we say that $\C$ has been {\it encoded} by the network $W$.  If  $\C = \P(W)$, so that there are no spurious states, then we say that $\C$
has been {\it exactly encoded} by $W$.

We tackle the NE problem by analyzing a novel learning rule, called the Encoding Rule.  In what follows, the problem is broken into four motivating questions that address (1) the learning rule, (2) the resulting structure of permitted sets, (3) binary codes that are exactly encodable, and (4) the structure of spurious states when codes are not encoded exactly.  In Section~\ref{sec:natural} we use our results to uncover ``natural'' codes for symmetric threshold-linear networks, and illustrate this phenomenon in the case of hippocampal PF codes in Section~\ref{sec:PFcodes}.

\subsection{The Encoding Rule}\label{sec:encoding-rule}

\q{1}{Is there a biologically plausible learning rule that allows arbitrary neural codes to be stored as permitted sets in threshold-linear networks?}
\medskip

In this section we introduce a novel encoding rule that constructs a network $W$ from a prescribed set of binary patterns $\C$. The rule is similar to the classical Hopfield learning rule \cite{Hopfield:1982} in that it updates the weights of the connectivity matrix $W$ following sequential presentation of binary patterns, and strengthens excitatory synapses between co-active neurons in the patterns.   In particular, the rule is Hebbian and {\it local}; i.e., each synapse is updated only in response to the co-activation of the two adjacent neurons, and the updates can be implemented by presenting only one pattern at a time \cite{Hopfield:1982,DayanAbbott}. 
A key difference from the Hopfield rule is that the synapses are {\it binary}: once a synapse $(ij)$ has been turned ``on,'' the value of $W_{ij}$ stays the same irrespective of the remaining patterns.\footnote{The learning rule in \cite{Seung:2002} also had binary synapses in this sense.}  A new ingredient is that synapses are allowed to be {\it heterogeneous}: in other words, the actual weights of connections are varied among ``on'' synapses.  These weights are assigned according to a predetermined {\it synaptic strength matrix} $S$, which is considered fixed and reflects the underlying architecture of the network.   For example, if no physical connection exists between neurons $i$ and $j$, then $S_{ij} = 0$ indicating that no amount of co-firing can cause a direct excitatory connection betwen those neurons.  On the other hand, if two neurons have multiple points of physical contact, then $S_{ij}$ will be greater than if there are only a few anatomical contacts.
 There is, in fact, experimental evidence in hippocampus for synapses that appear binary and heterogeneous in this sense \cite{Petersen:1998:PNAS}, with individual synapses exhibiting potentiation in an all-or-nothing fashion,  but having different ``thresholds'' for potentiation and heterogeneous synaptic strengths.
 
 Here we describe the Encoding Rule in general, with minimal assumptions on $S$.  Later, in Sections~\ref{sec:exact} and~\ref{sec:spurious}, we will investigate the consequences of various choices of $S$ on the network's ability to encode different types of binary neural codes.
\bigskip
\pagebreak

\ssec{Encoding Rule} This is a prescription for constructing (i.e., ``learning'') a network $W$ from a set of binary patterns on $n$ neurons, $\C \subset 2^{[n]}$ (e.g., $\C$ is a binary neural code).  It consists of three steps: two initialization steps (Steps 1-2), followed by an update step (Step 3). 

\begin{quote}
\noindent {\bf Step 1}: Fix an $n \times n$ synaptic strength matrix $S$ and an $\varepsilon>0$.  We think of $S$ and $\varepsilon$ as {\it intrinsic} properties of the underlying network architecture, established prior to learning.  
Because $S$ contains synaptic strengths for symmetric excitatory connections, we 
 require that $S_{ij} = S_{ji} \geq 0$ and $S_{ii}=0$.  
\medskip

\noindent {\bf Step 2}:  The network $W$ is initialized to be symmetric with effective connection strengths $W_{ij} = W_{ji} < -1$ for $i \neq j$, and $W_{ii} = 0$.  
(Beyond this requirement, the initial values of $W$ do not affect the results.)  
\medskip

\noindent {\bf Step 3}: Following presentation\footnote{By ``presentation'' of each pattern we mean that patterns are considered one at a time in building the $W$ matrix, without regard to the dynamics of equation~\eqref{eq:dynamics} (c.f. \cite{Hopfield:1982, Seung:2002}).}
 of each pattern $\sigma \in \C$, we turn ``on'' all excitatory synapses between neurons that co-appear in $\sigma$.  This means we update the relevant entries of $W$ as follows:
$$W_{ij} := -1 + \varepsilon S_{ij} \;\; \text{if} \; \; i,j \in \sigma  \text{ and }  i \neq j.$$
Note that the order of presentation does not matter; once an excitatory connection has been turned ``on,'' the value of $W_{ij}$ stays the same irrespective of remaining patterns.  
\end{quote}
\medskip

To better understand what kinds of networks result from the Encoding Rule, observe that any initial $W$ in Step 2 can be written as $W_{ij} = -1 - \varepsilon R_{ij}$, where $R_{ij} = R_{ji}>0$ for $i \neq j$ and $R_{ii} = -1/\varepsilon$, so that $W_{ii}=0$.
Assuming a threshold-linear network with homogeneous timescales, i.e. fixing $D = I$, 
the final network $W$ obtained from $\C$ after Step 3 satisfies,

\begin{equation}\label{eq:encoding-rule}
 (-D + W)_{ij} = \left\{\begin{array}{ccc} -1 + \varepsilon S_{ij}, & \text{if} & (ij) \in G(\C) \\ -1, &\text{if} & i = j \\ -1 - \varepsilon R_{ij} & \text{if} & (ij) \notin G(\C),\end{array}\right.
\end{equation}

\noindent where $G(\C)$ is the graph on $n$ vertices (neurons) having an edge for each pair of neurons that co-appear in one or more patterns of $\C$.  We call this graph the {\it co-firing graph} of $\C$.
In essence, the  rule allows the network to ``learn'' $G(\C)$, selecting which excitatory synapses are turned ``on'' and assigned to their predetermined weights.  

Consequently, any matrix $-D+W$ obtained via the Encoding Rule thus has the form 
$$-11^T+\varepsilon A,$$
where $-11^T$ denotes the $n \times n$ matrix of all $-1$s and 
$A$ is a symmetric matrix with zero diagonal and off-diagonal entries
$A_{ij} = S_{ij} \geq 0$ or $A_{ij} = -R_{ij} < 0$, depending on $\C$.
It then follows from Theorem~\ref{thm:thm1} that the permitted sets of this network are
$$\P(W) = \stab(-11^T + \varepsilon A).$$
Furthermore, it turns out that $\P(W)$ for {\it any} symmetric $W$ is of this form, even if $-D + W$ is not of the form $ -11^T+\varepsilon A.$

\begin{lemma} \label{lemma3}
 If $W$ is a symmetric threshold-linear network (with $D$ not necessarily equal to the identity matrix $I$), then there exists a symmetric $n \times n$ matrix $A$ with zero diagonal such that
$\P(W) = \stab(-11^T+A).$
\end{lemma}
\noindent The proof is given in Appendix B (Lemma~\ref{lemma:13}). 

In addition to being symmetric, the Encoding Rule (for small enough $\varepsilon$) generates ``lateral inhibition'' networks where the matrix $-D+W$ has strictly negative entries.  In particular, this means that the matrix $D-W$ is {\it copositive} -- i.e., $x^T(D-W)x>0$ for all nonnegative $x$ except $x = 0$.  
It follows from \cite[Theorem 1]{Seung:2003} that for all input vectors $b \in \RR^n$ and for all initial conditions, the network dynamics~\eqref{eq:dynamics} converge to an equilibrium point.  This was proven by constructing a Lyapunov-like function, similar to the strategy in \cite{cohen-grossberg:1983}.\footnote{Note that threshold-linear networks do not directly fall into the very general class of networks discussed in \cite{cohen-grossberg:1983}.}

\subsection{Main Result}\label{sec:main-result}

\q{2}{What is the full set of permitted sets $\P(W)$  stored in a network constructed using the Encoding Rule?}
\medskip

Our main result, Theorem~\ref{thm:main-result}, characterizes the full set of permitted sets $\P(W)$ obtained using the Encoding Rule, revealing a detailed understanding of the structure of spurious states.   Recall from Lemma~\ref{lemma3} that the set of permitted sets of any symmetric network on $n$ neurons has the form
$\P(W) = \stab(-11^T + \varepsilon A),$
for $\varepsilon>0$ and $A$ a symmetric $n \times n$ matrix with zero diagonal.\footnote{In fact, any $\P(W)$ of this form can be obtained by perturbing around any rank 1 matrix -- not necessarily symmetric -- having strictly negative diagonal (Proposition~\ref{prop:-1perturbation}, in Appendix B).}  
Describing $\P(W)$ thus requires understanding the stability of the principal submatrices $(-11^T + \varepsilon A)_\sigma$ for each $\sigma \subset [n]$.  Note that these submatrices all have the same form:  $-11^T+\varepsilon A_\sigma$,
where $-11^T$ is the all $-1$s matrix of size $|\sigma| \times |\sigma|$.  
Proposition~\ref{prop:main-result2} (below) provides an unexpected connection between the stability of these matrices and classical {\it distance geometry}.\footnote{Distance geometry is a field of mathematics that was developed in the early 20th century, motivated by the following problem:  find necessary and sufficient conditions such that a finite set of distances can be realized from a configuration of points in Euclidean space.
The classical text on this subject is Blumenthal's 1953 book \cite{Blumenthal}.}   
We will first present Proposition~\ref{prop:main-result2}, and then show how it leads us to Theorem~\ref{thm:main-result}. 

For symmetric $2 \times 2$ matrices of the form 
$-11^T + \varepsilon A = \left[\begin{array}{cc} -1 & -1 + \varepsilon A_{12} \\  -1 + \varepsilon A_{12} & -1 \end{array}\right]$, with $\varepsilon > 0$, it is easy to identify the conditions for the matrix to be stable.  Namely, one needs the determinant to be positive, so $A_{12}>0$ and $\varepsilon < 2/A_{12}$.  For $3 \times 3$ matrices, the conditions are more interesting, and the connection to distance geometry emerges.

\begin{lemma} \label{lemma:3by3}
Consider the $3 \times 3$ matrix $-11^T+\varepsilon A$, for a fixed symmetric $A$ with zero diagonal:
$$\left[\begin{array}{ccc} 
-1 & -1+ \varepsilon A_{12} & -1 + \varepsilon A_{13} \\ 
 -1+ \varepsilon A_{12} & -1 & -1 + \varepsilon A_{23} \\
 -1+ \varepsilon A_{13} & -1+\varepsilon A_{23} & -1
 \end{array}\right].$$
There exists an $\varepsilon>0$ such that this matrix is stable if and only if 
$\sqrt{A_{12}}, \sqrt{A_{13}},$ and $\sqrt{A_{23}}$ are valid edge lengths for a nondegenerate triangle in $\RR^2$.
\end{lemma}

\noindent In other words, the numbers $\sqrt{A_{ij}}$ must satisfy the triangle inequalities $\sqrt{A_{ij}} < \sqrt{A_{ik}} + \sqrt{A_{jk}}$ for distinct $i,j,k$.
This can be proven by straightforward computation, using Heron's formula and the characteristic polynomial of the matrix.  The upper bound on $\varepsilon$, however, is not so easy to identify.

Remarkably, the above observations completely generalize to $n \times n$ matrices of the form $-11^T + \varepsilon A$, and the precise limits on $\varepsilon$ can also be computed for general $n$.  This is the content of Proposition~\ref{prop:main-result2}, below.  To state it, however, we first need a few notions from distance geometry.

\begin{definition}
An $n \times n$ matrix $A$ is a (Euclidean) {\em square distance matrix} if there exists a configuration of points $p_1,...,p_n \in \RR^{n-1}$ (not necessarily distinct) such that $A_{ij} = \norm{p_i - p_j}^2$.   
$A$ is a {\em nondegenerate} square distance matrix if the corresponding points are affinely independent; i.e., if the convex hull of $p_1,...,p_n$ is a simplex with nonzero
volume in $\RR^{n-1}$. 
\end{definition}

\noindent Clearly, all square distance matrices are symmetric and have zero diagonal.
Furthermore, a $2 \times 2$ matrix $A$ is a {\it nondegenerate} square distance matrix if and only if the off-diagonal entry satisfies the additional condition $A_{12}>0$.  For a $3 \times 3$ matrix $A$, the necessary and sufficient condition to be a nondegenerate square distance matrix is that the entries $\sqrt{A_{12}}, \sqrt{A_{13}},$ and $\sqrt{A_{23}}$ are valid edge lengths for a nondegenerate triangle in $\RR^2$ (this was precisely the condition in Lemma~\ref{lemma:3by3}).  For larger matrices, however, the conditions are less intuitive.  A key object for determining whether or not and $n \times n$ matrix $A$ is a nondegenerate square distance matrix is the
{\it Cayley-Menger determinant}: 
$$\CM(A)\od \det \left[\begin{array}{cc} 0 & 1^T \\ 1 &A\end{array}\right],$$ 
where $1 \in \RR^n$ is the column vector of all ones. 
If $A$ is a square distance matrix, then $\CM(A)$ is proportional to the square volume of the simplex obtained as the convex hull of the points $\{p_i\}$  (see Lemma~\ref{lemma:cayley-menger} in Appendix A).  In particular, $\CM(A) \neq 0$ (and hence $|\CM(A)|>0$) if $A$ is a {\it nondegenerate} square distance matrix, while $\CM(A) = 0$ for any other (degenerate) square distance matrix.  

%With these preliminaries we can now state Proposition~\ref{prop:main-result2}, a suprising fact that generalizes our observations for $2 \times 2$ and $3 \times 3$ matrices.

\begin{proposition} \label{prop:main-result2}
Let $A$ be a symmetric $n \times n$ matrix with zero diagonal, and $\varepsilon > 0$.  Then the matrix
$$-11^T+\varepsilon A$$  is stable if and only if the following two conditions hold: 
\begin{enumerate}
\item[(a)] $A$ is a nondegenerate square distance matrix, and
\item[(b)] $0<\varepsilon < |\CM(A)/\det(A)|.$
\end{enumerate}
\end{proposition}
\noindent  Proposition~\ref{prop:main-result2} is essentially a special case of Theorem~\ref{thm:perturbation} -- our core technical result -- whose statement and proof are given in Section~\ref{sec:thm4}.  The proof of Proposition~\ref{prop:main-result2} is then given in Section~\ref{sec:thm2-proof}. 
To our knowledge, Theorem~\ref{thm:perturbation} is novel and connections to distance geometry have not previously been used in the study of neural networks or, more generally, the stability of fixed points in systems of ODEs.

The ratio $|\CM(A)/\det(A)|$ has a simple geometric interpretation in cases where condition (a) of Proposition~\ref{prop:main-result2} holds.  Namely, if $A$ is an $n \times n$ nondegenerate square distance matrix (with $n > 1$), then $|\CM(A)/\det(A)| = \dfrac{1}{2\rho^2},$
where $\rho$ is the radius of the unique sphere circumscribed on any set of points in Euclidean space that can be used to generate $A$ 
(see Remark 1 in Appendix C).
Moreover, since $|\CM(A)|>0$ whenever $A$ is a nondegenerate square distance matrix, there always exists an $\varepsilon$ small enough to satisfy the second condition, provided the first condition holds.  
Combining Proposition~\ref{prop:main-result2} together with Cauchy's interlacing theorem yields:

\begin{lemma}\label{lemma:inequality}
If $A$ is an $n \times n$ nondegenerate square distance matrix, then 
$$ 0 < \left| \dfrac{\CM(A_\sigma)}{\det(A_\sigma)} \right| \leq \left| \dfrac{\CM(A_\tau)}{\det(A_\tau)} \right| \;\; \text{if} \;\; \tau \subseteq \sigma \subseteq [n].$$
\end{lemma}

Given any symmetric $n\times n$ matrix $A$ with zero diagonal, and $\varepsilon > 0$, it is now natural to define the following simplicial complexes:
\begin{eqnarray*}
\geom_{\varepsilon}(A) &\od& \{ \sigma \subseteq [n] \mid A_\sigma \text{ a nondeg. sq. dist. matrix}
\text{ and } \left|\dfrac{\CM(A_\sigma)}{\det(A_\sigma)}\right| > \varepsilon \}, \; \text{and}\;\\
\geom(A) &\od& \lim_{\varepsilon \rightarrow 0} \geom_{\varepsilon}(A)
= \{ \sigma \subseteq [n] \mid A_\sigma \text{ a nondeg. sq. dist. matrix}\}.
\end{eqnarray*}
Lemma~\ref{lemma:inequality} implies
that $\geom_{\varepsilon}(A)$ and $\geom(A)$ are simplicial complexes.
Note that if $\sigma = \{i\}$, we have $A_\sigma = [0]$.  In this case, $\{i\} \in \geom(A)$ and $\{i\} \in \geom_\varepsilon(A)$ for all $\varepsilon>0$, by our convention.  
Also,
$\geom_\varepsilon(A) = \geom(A)$ if and only if  $0 < \varepsilon < \delta(A)$, where 
$$
\delta(A) \od \min \left\{\left| \dfrac{\CM(A_\sigma)}{\det(A_\sigma)} \right|\right\}_{\sigma \in \geom(A)}.
$$
If $A$ is a nondegenerate square distance matrix, then $\delta(A) = |\CM(A)/\det(A)|$.
 
To state our main result, Theorem~\ref{thm:main-result}, we also need a few standard notions from graph theory.
A {\em clique} in a graph $G$ is a subset of vertices that is all-to-all connected.\footnote{For recent work encoding cliques in Hopfield networks, see \cite{hillar:2012}.}  
The {\em clique complex} of $G$, denoted $X(G)$, is the set of all cliques in $G$;
this is a simplicial complex for any $G$.  Here we are primarily interested in the graph $G(\C)$, the co-firing graph of a set of binary patterns $\C \subset 2^{[n]}$.

\begin{theorem}\label{thm:main-result}
Let $S$ be an $n \times n$ synaptic strength matrix satisfying $S_{ij}=S_{ji} \geq 0$ and $S_{ii}=0$ for all $i,j \in [n]$, and fix $\varepsilon > 0$.
Given a set set of prescribed patterns $\C \subset 2^{[n]}$, let $W$ be the threshold-linear network (equation~\eqref{eq:encoding-rule})
 obtained from $\C$  using $S$ and $\varepsilon$ in the Encoding Rule. 
Then, $$\P(W) = \geom_\varepsilon(S) \cap X(G(\C)).$$
If we further assume that $\varepsilon < \delta(S)$, then $\P(W) = \geom(S) \cap X(G(\C)).$

In other words, a binary pattern $\sigma \subset [n]$ is a permitted set of $W$
if and only if $S_\sigma$ is a nondegenerate square distance matrix, $\varepsilon < |\CM(S_\sigma)/\det(S_\sigma)|$, and $\sigma$ is 
a clique in the graph $G(\C)$.
\end{theorem}

\noindent The proof is given in Section~\ref{sec:thm2-proof}.   Theorem~\ref{thm:main-result} answers Question 2, and makes explicit how $\P(W)$ depends on 
$S$, $\varepsilon$, and $\C$.  One way of interpreting this result is to observe that a binary pattern $\sigma \in \C$ is successfully stored as a permitted set of $W$
if and only if the excitatory connections between the neurons in $\sigma$, given by $\tilde{S}_\sigma = \varepsilon S_{\sigma}$, are {\it geometrically balanced} -- i.e.,
\begin{itemize}
\item $\tilde{S}_\sigma$ is a nondegenerate square distance matrix, and
\item $|\det(\tilde{S}_\sigma)| < |\CM(\tilde{S}_\sigma)|$.
\end{itemize}
The first condition ensures a certain balance among the relative strengths of excitatory connections in the clique $\sigma$, while the second condition bounds
the overall excitation strengths relative to inhibition (which has been normalized to $-1$ in the Encoding Rule).

We next turn to an example that illustrates
how this theorem can be used to solve the NE problem explicitly for a small binary neural code.  In the following section, Section~\ref{sec:exact}, we address more generally the question of what neural codes can be encoded exactly, and what is the structure of spurious states when a code is encoded inexactly.

\subsection{An example} \label{sec:example}
Suppose $\C$ is a binary neural code on $n = 6$ neurons, consisting of 
%following $22$ patterns:
%\begin{eqnarray*}
%\C &=& \{110100, 101010, 011001, 000111, 110000, 101000, 100100, 100010, 011000, 010100, 010001,\\ 
%&& 001010, 001001, 000110, 000101, 000011, 100000, 010000, 001000, 000100, 000010, 000001\}.
%\end{eqnarray*}
maximal patterns 
$$\{110100, 101010, 011001, 000111\},$$
corresponding to subsets $\{124\}, \{135\}, \{236\},$ and $\{456\}$, 
together with all subpatterns (smaller subsets) of the maximal ones, thus ensuring that $\C$ is a simplicial complex.
This is depicted in Figure 2A, using a standard method of illustrating simplicial complexes geometrically.  The four maximal patterns correspond to the shaded triangles, while patterns with only one or two co-active neurons comprise the vertices and edges of the co-firing graph $G(\C)$.\footnote{In this example, there are no patterns having four or more neurons, but these would be illustrated by tetrahedra and other higher-order simplices.}

\begin{figure}[!h]  \label{fig:example}
\begin{center}
%\vspace{-.2in}
\includegraphics[width=6.5in]{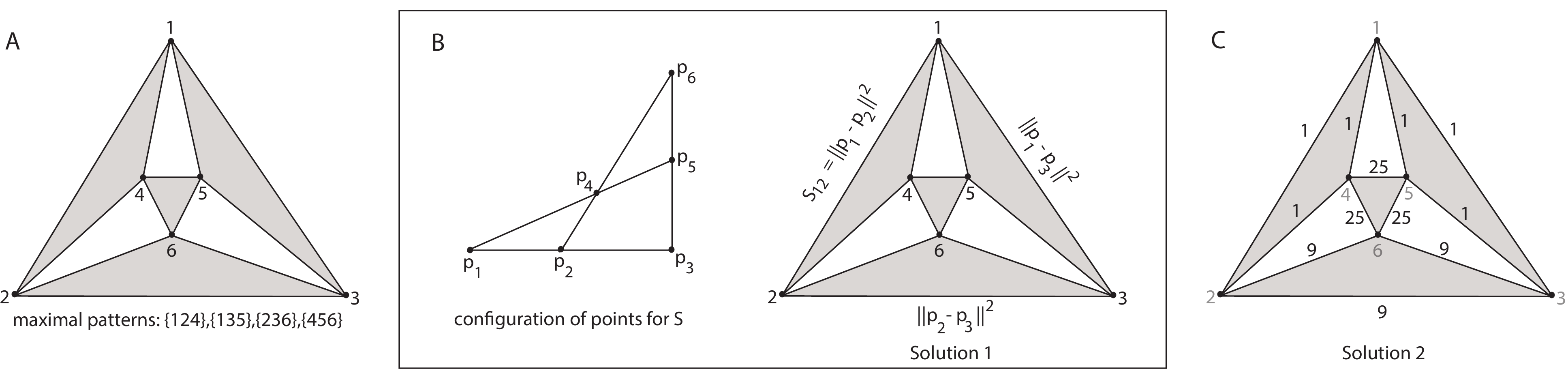}
\end{center}
\vspace{-.15in}
\caption{\small An example on $n = 6$ neurons. (A) The simplicial complex $\C$ consists of 4 two-dimensional facets (shaded triangles).  The graph $G(\C)$ contains the 6 vertices and 12 depicted edges; these are also included in $\C$, so the size of the code is $|\C| = 22$.
(B) A configuration of points $p_1,...,p_6 \in \RR^2$ that can be used to exactly encode $\C$.
Lines indicate triples of points that are collinear.  From this configuration we construct a $6 \times 6$ synaptic strength matrix $S$, with $S_{ij} = \norm{p_i-p_j}^2,$ and choose $0<\varepsilon<\delta(S)$.   
The geometry of the configuration implies that
$\geom(S)$ does not contain any patterns of size greater than $3$, nor does it contain
the triples $\{1 2 3\}, \{1 4 5\}, \{2 4 6\},$ or $\{3 5 6\}$. It is straightforward to check that $\C = \geom(S) \cap X(G(\C))$.
(C) Another solution for exactly encoding $\C$ is provided by choosing the matrix $S$ with $S_{ij}$ given by the labeled edges in the figure.  The square distances in $S_{ij}$ were chosen to satisfy the triangle inequalities for shaded triangles, but to violoate them for empty triangles.}
\end{figure}  

Without Theorem~\ref{thm:main-result}, it is difficult to find a network $W$ that encodes $\C$ exactly -- i.e., such that $\P(W) = \C$.  This is in part because each connection strength $W_{ij}$ belongs to two $3 \times 3$ matrices that must satisfy opposite stability properties.  For example, the subset $\{124\}$ must be a permitted set of $\P(W)$ while $\{123\}$ is not permitted, imposing competing conditions on the entry $W_{12}$.  In general, it may be difficult to ``patch'' together local ad-hoc solutions to obtain a single matrix $W$ having all the desired stability properties.  

Using Theorem~\ref{thm:main-result}, on the other hand, we can easily construct many exact solutions for encoding $\C$ as a set of permitted sets $\P(W)$.  The main idea is as follows.  Consider the Encoding Rule with synaptic strength matrix $S$ and $0<\varepsilon < \delta(S)$.  Applying the rule to $\C$ yields a network with permitted sets 
$$\P(W) = \geom(S) \cap X(G(\C)).$$
The goal is thus to find $S$ so that $\C = \geom(S) \cap X(G(\C)).$
From the co-firing graph $G(\C)$, we see that the clique complex $X(G(\C))$ contains all triangles depicted in Figure 2A, including the ``empty'' (non-shaded) triangles: $\{1 2 3\}, \{1 4 5\}, \{2 4 6\},$ and $\{3 5 6\}$.
The matrix $S$ must therefore be chosen so that these triples are not in $\geom(S)$, while ensuring that $\{124\}, \{135\}, \{236\},$ and $\{456\}$ are included.  In other words, to obtain an exact solution we must find $S$ such that $S_\sigma$ is a nondegenerate square distance matrix for each $\sigma \in \{\{124\}, \{135\}, \{236\},\{456\} \}$, but {\it not} for $\sigma$ corresponding to an empty triangle.
\medskip

\ssec{Solution 1}  Consider the configuration of points $p_1,...,p_6 \in \RR^2$ in Figure 2B, and let $S$ be the $6 \times 6$ square distance matrix with entries $S_{ij} = \norm{p_i-p_j}^2.$  Because the points lie in the plane, the largest principal submatrices of $S$ that can possibly be nondegenerate square distance matrices are $3 \times 3$.  This means $\geom(S)$ has no elements of size greater than $3$.  Because no two points have the same position, $\geom(S)$ contains the complete graph with all edges $(ij)$.  It remains only to determine which triples are in $\geom(S)$.  The only $3 \times 3$ principal submatrices of $S$ that are nondegenerate square distance matrices correspond to triples of points in general position.  From Figure 2B (left) we see that $\geom(S)$ includes all triples {\it except} $\{1 2 3\}, \{1 4 5\}, \{2 4 6\},$ and $\{3 5 6\}$, since these correspond to triples of points that are collinear (and thus yield {\it degenerate} square distance matrices).  Although $\C \neq X(G(\C))$ and $\C \neq \geom(S)$, it is now easy to check that $\C = \geom(S) \cap X(G(\C))$.  Using Theorem~\ref{thm:main-result}, we can conclude that $\C = \P(W)$ exactly, where $W$ is the network obtained using the Encoding Rule with this $S$ and any $0< \varepsilon < \delta(S)$.\medskip

\ssec{Solution 2} Let $S$ be the symmetric matrix defined by the following equations for $i < j$: $S_{ij} = 1$ if $i=1$; $S_{24} = S_{35} = 1$; $S_{23} = S_{26} = S_{36} = 3^2$; and $S_{ij} = 5^2$ if $i = 4$ or $5$.  Here we've only assigned values corresponding to each edge in $G(\C)$ (see Figure 2C); remaining entries may be chosen arbitrarily, as they play no role after we intersect $\geom(S) \cap X(G(\C))$.
Note that $S$ is {\it not} a square distance matrix at all, not even a degenerate one.  Nevertheless, $S_\sigma$ {\it is} a nondegenerate square distance matrix for $\sigma \in \{\{124\}, \{135\}, \{236\},\{456\} \}$, because the distances correspond to nondegenerate triangles.  For example, the triple $\{124\}$ has pairwise distances $(1, 1, 1)$, which satisfy the triangle inequality.  In contrast, the triple $\{123\}$ has pairwise distances $(1, 1, 3)$, which violate the triangle inequality; hence, $S_{\{123\}}$ is {\it not} a square distance matrix.  Similarly,  the triangle inequality is violated for each of $\{1 4 5\}, \{2 4 6\},$ and $\{3 5 6\}$.
It is straightforward to check that, among all cliques of $X(G(\C))$, only the desired patterns in $\C$ are also elements of $\geom(S)$,  so $\C = \geom(S) \cap X(G(\C))$.\medskip

By construction, Solution 1 and Solution 2 produce networks $W$ (obtained using the Encoding Rule with $\varepsilon, S$ and $\C$) with exactly the same set of permitted sets $\P(W)$.  Nevertheless, the solutions are functionally different in that the resulting input-output relationships associated to the equation~\eqref{eq:dynamics} dynamics are different, as they depend on further details of $W$ not captured by $\P(W)$ (see Appendix D).

\subsection{Binary neural codes that can be encoded exactly}\label{sec:exact}

\q{3}{What binary neural codes can be encoded exactly as $\C = \P(W)$ for a symmetric threshold-linear network $W$?}
\medskip

\q{4}{If encoding is not exact, what is the structure of spurious states?}
\medskip

From Theorem~\ref{thm:main-result}, it is clear that if the set of patterns to be encoded happens to be of the form $\C = \geom(S) \cap X(G(\C))$, then $\C$ can be exactly encoded as $\P(W)$ for small enough $\varepsilon$ and the same choice of $S$.  Similarly, if the set of patterns has the form $\C = \geom_\varepsilon(S) \cap X(G(\C))$, then $\C$ can be exactly encoded as $\P(W)$ using our Encoding Rule (Section~\ref{sec:encoding-rule}) with the same $S$ and $\varepsilon$.  Can any other sets of binary patterns be encoded exactly via symmetric threshold-linear networks?  The next theorem assures us that the answer is ``No.''  This means that by focusing attention on networks constructed using our Encoding Rule, we are not missing any binary neural codes that could arise as $\P(W)$ for other symmetric networks.

\begin{theorem}\label{thm:main-result2}
Let $\C \subset 2^{[n]}$ be a binary neural code.  There exists a symmetric threshold-linear network $W$ such that
$\C = \P(W)$ if and only if
$\C$ is a simplicial complex of the form
\begin{equation}\label{eq:P-form}
\hspace{.75in} \C = \geom_\varepsilon(S) \cap X(G(\C)),
\end{equation}
for some $\varepsilon > 0$ and $S$ an $n \times n$ matrix satisfying $S_{ij} = S_{ji} \geq 0$ and $S_{ii} = 0$ for all $i,j \in [n]$.  Moreover, $W$ can be constructed using the Encoding Rule on $\C$, using this choice of $S$ and $\varepsilon$.
\end{theorem}

\noindent The proof is given in Section~\ref{sec:thm2-proof}.
Theorem~\ref{thm:main-result2} allows us to make a preliminary classification of binary neural codes that can be encoded exactly,
giving a partial answer to Question 3.
To do this, it is useful to distinguish three different types of $S$ matrices that can be used in the Encoding Rule.  

\begin{table}
\begin{center}
\begin{tabular}{c|c|c|c}
type of & $\C$ that can be {\it exactly}  & $\C$ that can be & spurious states\\
$S$ matrix & encoded: $\C = \P(W)$ & encoded: $\C \subset \P(W)$ &  $\P(W) \setminus \C$\\
\hline
\hline
&&&\\
universal $S$ & any clique complex & all codes & cliques of $G(\C)$\\
& $X(G)$ & & that are not in $\C$\\
&&&\\
\hline
&&&\\
$k$-sparse & any $k$-skeleton $X_k(G)$ & all $k$-sparse codes & cliques of $G(\C)$ of size $\leq k$,\\
universal $S$ & of a clique complex & ($|\sigma| \leq k$ for all $\sigma \in \C$)&  that are not in $\C$ \\
&&&\\
\hline
&&&\\
specially-  & $\C$ is of the form & depends on $S$ & cliques of $G(\C)$ that are\\
tuned $S$ & $\geom(S) \cap X(G)$ & & in $\geom(S)$ but not in $\C$\\
&&&\\
\hline
\end{tabular}
\caption{\small Classification of $S$ matrices, together with encodable codes and spurious states.
Note: the above assumes using the Encoding Rule on the code $\C$ with synaptic strength matrix
$S$ and $0<\varepsilon<\delta(S)$.  Additional codes may be exactly
encodable for other choices of $\varepsilon$.}
\end{center}
\end{table}

\begin{itemize}
\item {\bf universal $S$}.  We say that a matrix $S$ is ``universal'' if it is an $n \times n$ nondegenerate square distance matrix.
In particular, any principal submatrix $S_\sigma$ is also a nondegenerate square distance matrix, so if we
 let $0 < \varepsilon < \delta(S) = \left|\CM(S)/\det(S)\right|$ then any $\sigma \in \C$ has corresponding excitatory connections
 $\varepsilon S_\sigma$ that are geometrically balanced (see Section~\ref{sec:main-result}).  Furthermore, $\geom_\varepsilon(S) = \geom(S) = 2^{[n]}$,  and hence $\geom_\varepsilon(S) \cap X(G(\C)) = X(G(\C))$, irrespective of $S$.  It follows that if $\C = X(G)$ for any graph $G$, 
 then $\C$ can be {\it exactly} encoded using any universal $S$ and any $0<\varepsilon<\delta(S)$ in the Encoding Rule.\footnote{Note that if $\C = X(G)$ is any clique complex with underlying graph $G$, then we automatically know that $G(\C) = G$, and hence $X(G(\C)) = X(G) = \C$.} Moreover, since $\C \subset X(G(\C))$ for any code $\C$, it follows that any code can be encoded -- albeit inexactly -- using
 a universal $S$ in the Encoding Rule.  Finally, the spurious states $\P(W) \setminus \C$ can be completely understood: they consist of all cliques in the graph $G(\C)$ that are not elements of $\C$.

\item {\bf $k$-sparse universal $S$}.  We say that a matrix $S$ is ``$k$-sparse universal''  if it is a (degenerate) $n \times n$ square distance matrix for a configuration of $n$ points that are 
in {\it general position}\footnote{This guarantees that all $k \times k$ principal submatrices of $S$ are nondegenerate square distance matrices.} in $\RR^{k},$ for $k < n-1$ (otherwise $S$ is universal).    Let $0<\varepsilon<\delta(S)$.  
Then, $\geom_\varepsilon(S) = \geom(S) =  \{\sigma \subset [n] \mid |\sigma| \leq k+1 \}$; this is
 the {\it $k$-skeleton}\footnote{The $k$-skeleton $\Delta_k$ of a simplicial complex $\Delta$ is obtained by restricting to faces of dimension $\leq k$, which corresponds to keeping only elements $\sigma \subset \Delta$ of size $|\sigma| \leq k+1$.  Note that $\Delta_k$ is also a simplicial complex.} of the complete simplicial complex $2^{[n]}$.  
 This implies that $\geom_\varepsilon(S) \cap X(G(\C)) = X_k(G(\C))$, where $X_k$ denotes the $k$-skeleton of the clique complex $X$:
 $$X_k(G(\C)) \od \{ \sigma \in X(G(\C)) \mid |\sigma| \leq k+1 \}.$$
It follows that any $k$-skeleton of a clique complex, $\C = X_k(G)$ for any graph $G$, can be encoded exactly.
Furthermore, since any $k$-sparse code $\C$ satisfies $\C \subseteq X_k(G(\C))$, any $k$-sparse code can be encoded using this type of $S$ matrix in the Encoding Rule.  The spurious states in this case are cliques of $G(\C)$ that have size no greater than $k$.

\item {\bf specially-tuned $S$}.  We will refer to all $S$ matrices that do not fall into the universal or $k$-sparse universal categories as ``specially-tuned.''  In this case, we cannot say anything general about the codes that are exactly encodable without further knowledge about $S$.  If we let $0<\varepsilon<\delta(S)$, as above, Theorem~\ref{thm:main-result2} tells us that the binary codes $\C$ can be encoded exactly (via the Encoding Rule) are of the form $\C = \geom(S) \cap X(G(\C))$.  Unlike in the universal and $k$-sparse universal cases, the encodable codes depend on the precise form of $S$.
Note that the example code $\C$ discussed in Section~\ref{sec:example} was not a clique complex nor the $k$-skeleton of a clique complex.  Nevertheless, it could be encoded exactly for the ``specially-tuned" choices of $S$ exhibited in Solution 1 and Solution 2 (see Figure 2B,C).
\end{itemize}

\noindent A summary of what codes are encodable and exactly encodable for each type of $S$ matrix is shown in Table 2, under the assumption that $0<\varepsilon<\delta(S)$ in the Encoding Rule.

We end this section with several technical Remarks, along with some open questions for further mathematical investigation.

\subsection*{Remarks}

\noindent{\bf 1. Fine-tuning?}
It is worth noting here that solutions obtained by choosing $S$ to be a {\it degenerate} square distance matrix, as in the
 $k$-sparse universal $S$ or the specially-tuned $S$ of Figure 2B,
   are not as fine-tuned as they might first appear.  This is because the ratio $|\CM(S_\sigma)/\det(S_\sigma)|$ approaches zero as subsets of points $\{p_i\}_{i\in\sigma}$ 
used to generate $S$ become {\it approximately} degenerate, allowing elements to be eliminated from $\geom_\varepsilon(S)$ because of violations to condition (b) in Proposition~\ref{prop:main-result2}, even
if condition $(a)$ is not quite violated. This means the appropriate matrices do not have to be exactly degenerate, but only approximately degenerate (see Remark 2 in Appendix C).  In particular, the collinear points in Figure 2B need not be exactly collinear for Solution 1 to hold.
\bigskip

\ssec{2. Controlling spurious cliques in sparse codes}
If the set of patterns $\C \subset 2^{[n]}$ to be encoded is a {\it $k$-sparse} code, i.e. if $|\sigma| \leq k < n$ for all $\sigma \in \C$, then any clique of size $k+1$ or greater
in $G(\C)$ is potentially a spurious clique.  We can eliminate these spurious states, however, by choosing 
a $k$-sparse universal $S$ in the Encoding Rule.  This guarantees that $\geom_\varepsilon(S)$ does not include any element of size greater than $k$, and hence $\P(W) \subseteq X_{k-1}(G(\C))$.  
\bigskip

\ssec{3. Uniform $S$} 
To use truly binary synapses, we can choose $S$ in the Encoding Rule to be the uniform synaptic strength matrix having $S_{ij} = 1$ for $i \neq j$ and $S_{ii} = 0$ for all $i \in [n]$.  
In fact, $S$ is a nondegenerate square distance matrix, so this is a special case of a ``universal'' S.
Here $\delta(S)$ turns out to have a very simple form: 
 $$\delta(S) = \left|\CM(S)/\det(S)\right| = \dfrac{n}{n-1}.$$ 
Similarly, any $k \times k$ principal submatrix $S_\sigma$, with $|\sigma| = k$, satisfies
$\delta(S_\sigma) = \dfrac{k}{k-1}.$  This implies that $\geom_\varepsilon(S)$ is the $k$-skeleton of the complete simplicial complex on $n$ vertices if and only if
$$\dfrac{k+2}{k+1} < \varepsilon < \dfrac{k+1}{k}.$$  It follows that for this choice of $S$ and $\varepsilon$ (note that $\varepsilon > \delta(S))$, the Encoding Rule yields $\P(W) = X_k(G(\C))$, just as in the case of $k$-sparse universal $S$.  If, on the other hand, we choose $0< \varepsilon \leq 1 < \delta(S)$, then $\geom_\varepsilon(S) = \geom(S) = 2^{[n]}$,
then we have the usual properties for universal $S$ (c.f. \cite{Seung:2002}).
\bigskip

 \ssec{4. Matroid complexes} In the special case where $S$ is a square distance matrix, $\geom(S)$ is a 
 {\it representable matroid complex} -- i.e., it is the independent set complex of
a real-representable matroid \cite{Oxley2011}.  Moreover, all representable matroid complexes
are of this form, and can thus be encoded exactly.  To see this, take any code $\C$ having $G(\C) = K_n$, the complete graph on $n$ vertices.  Then $X(G(\C)) = 2^{[n]}$, and the Encoding Rule (for $\varepsilon < \delta(S)$) yields 
$$\P(W) = \geom(S).$$
Note that although the example code $\C$ of Section~\ref{sec:example} is not a matroid complex (in particular, it violates the independent set exchange property \cite{Oxley2011}), $\geom(S)$ for the matrix $S$ given in Solution 1 (Figure 2B) {\it is} a representable matroid complex, showing that $\C$ is the intersection of a representable matroid complex and the clique complex $X(G(\C))$. 
\bigskip

\ssec{\bf 5. Open questions}
Can a combinatorial description be found for all simplicial complexes that are of the form $\geom_\varepsilon(S)$ or $\geom(S)$, where $S$ and $\varepsilon$ satisfy the conditions in Theorem~\ref{thm:main-result2}?  For such complexes, can the appropriate $S$ and $\varepsilon$ be obtained constructively?  Does every simplicial complex $\C$ admit an exact solution to the NE problem via a {\it symmetric} network $W$? 
I.e., is every simplicial complex of the form $\geom_\varepsilon(S) \cap X(G(\C))$, as in equation~ \eqref{eq:P-form}?
If not, what are the obstructions?
More generally, does every simplicial complex admit an exact solution (not necessarily symmetric) to the NE problem? We have seen that all matroid complexes for representable matroids can be exactly encoded as $\geom(S)$.  Can non-representable matroids also be exactly encoded?

\subsection{Spurious states and ``natural'' codes}\label{sec:spurious}

Although it may be possible, as in the example of Section~\ref{sec:example}, to precisely tune the synaptic strength matrix $S$ to exactly encode a particular neural code, this is somewhat contrary to the spirit of the Encoding Rule, which  assumes $S$ to be an intrinsic property of the underlying network.  
Fortunately, as seen in Section~\ref{sec:exact}, Theorem~\ref{thm:main-result} implies that certain ``universal'' choices of $S$ enable any $\C \subset 2^{[n]}$ to be encoded.  The price to pay, however, is the emergence of spurious states.

Recall that spurious states are permitted sets that arise in $\P(W)$ that  were {\it not} in the prescribed list $\C$ of binary patterns to be encoded.  Theorem~\ref{thm:main-result} immediately implies that all spurious states lie in $X(G(\C))$ -- i.e., every spurious state is a clique of the co-firing graph $G(\C)$.
We can divide them into two types: 

\begin{itemize}
\item {\bf Type 1: spurious subsets.}  These are permitted sets $\sigma \in \P(W) \setminus \C$ that are subsets of patterns in $\C$.  Note that if $\C$ is a simplicial complex, there will not be any spurious states of this type.  On the other hand, if $\C$ is {\it not} a simplicial complex, then type 1 spurious states are guaranteed to be present for any symmetric encoding rule, because $\P(W) = \stab(-D+W)$ is a simplicial complex for symmetric $W$ (Lemma~\ref{lemma:simplicial}).

\item {\bf Type 2: spurious cliques.}  These are permitted sets $\sigma \in \P(W) \setminus \C$ that are {\it not} of the first type.  Note that, technically, the type 1 spurious states are also cliques in $G(\C)$, but we will use the term ``spurious clique'' to refer only to spurious states that are not spurious subsets. \end{itemize}

Perhaps surprisingly, some common neural codes have the property that the full set of patterns to be encoded naturally contains a large fraction of the cliques in the code's co-firing graph.  In such cases,  
$\C \approx X(G(\C)),$ or $\C \approx X_k(G(\C))$.  These neural codes therefore have very few spurious states
when encoded using a universal or $k$-sparse universal $S$, even though $S$ has not been specially-tuned for the given code.  We will refer to these as {\it natural} codes for symmetric threshold-linear networks, because they have two important properties that make them particularly fitting for these networks:
\begin{enumerate}
\item[P1.] Natural codes can be encoded exactly or nearly-exactly, using any universal or k-sparse universal 
matrix $S$ in the Encoding Rule, and 
\item[P2.] Natural codes can be fully encoded following presentation of only a small (randomly-sampled) fraction of the patterns in the code.
\end{enumerate}
In other words, not only can natural codes be generically encoded with very few spurious states, but they can also be encoded from a highly undersampled set of codewords.  This is because the network naturally ``fills in'' the missing elements via spurious states that emerge after encoding only part of the code.
In the next two sections, we will explain why RF codes are ``natural'' in this sense, and illustrate the above two properties with a concrete application of encoding two-dimensional PF codes, an important example of RF codes.

\subsection{Receptive field codes are natural codes}\label{sec:natural}

\begin{wrapfigure}{r}{.4\linewidth}
\begin{center}
\vspace{-.35in}
\includegraphics[width=2in]{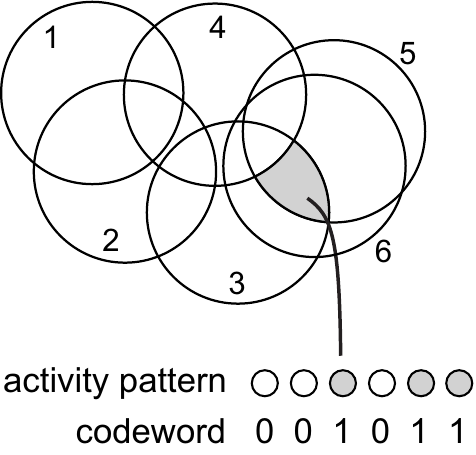}
\end{center}
%\vspace{-.15in}
\caption{\small Two-dimensional receptive fields for $6$ neurons.  The RF code $\C$ has a codeword
for each overlap region.  For example, the shaded region corresponds to the binary pattern $001011$; equivalently, we denote it as $\sigma = \{3,5,6\} \in \C$.  The corresponding {\it coarse} RF code also includes all subsets, such
as $\tau = \{3,5\},$ even if they are not part of the original RF code.}
\end{wrapfigure}

RF codes are binary neural codes consisting of activity patterns of populations of neurons that fire according to {\it receptive fields}.\footnote{In the vision literature, 
the term ``receptive field'' is reserved for subsets of the visual field; we use the term in a more general sense, applicable to any modality.}
Abstractly, a receptive field is a map $f_i:\S \rightarrow \RR_{\geq 0}$ from a space of stimuli $\S$ to the average firing rate $f_i(s)$ of a single neuron $i$ in response to each stimulus $s \in \S$.   Receptive fields are computed from experimental data by correlating neural
responses to external stimuli.  We follow a common abuse of language, where both the map and its support (i.e., the subset $U_i \subset \S$ where $f_i$ takes on
strictly positive values) are referred to as ``receptive fields.''  
If the stimulus space is $d$-dimensional, i.e. $\S \subset \RR^d$, we say that the receptive fields have {\it dimension} $d$. The paradigmatic examples of neurons with receptive fields are orientation-selective neurons in visual cortex \cite{Ben-Yishai1995} and 
hippocampal place cells \cite{PathIntegration}.  Orientation-selective neurons have {\it tuning curves} that reflect a neuron's preference for a particular angle. Place cells are neurons that have {\it place fields}  \cite{O:1976,ON:1978}; i.e., each neuron has a preferred (convex) region of the animal's physical environment where it has a high firing rate.  Both tuning curves and place fields are examples of low-dimensional receptive fields, having typical dimension $d=1$ or $d=2$.

The elements of a RF code $\C$ correspond to subsets of neurons that may be co-activated in response to a stimulus $s \in \RR^d$ (see Figure 3).
Here we define two variations of this notion, which we refer to as RF codes and {\it coarse} RF codes.

\begin{definition}
Let $\{U_1,...,U_n\}$ be a collection of convex open sets in $\RR^d$, where each $U_i$ is the receptive field corresponding to the $i$th neuron.  To such a set of receptive fields, we associate a $d$-dimensional {\it RF code} $\C$, defined as follows: 
for each $\sigma \in 2^{[n]}$, 
$$\sigma \in \C \text{\;\; if and only if \;\;} \bigcap_{i \in \sigma} U_i \setminus \bigcup_{j \notin \sigma} U_j \neq \emptyset.$$   
  This definition was previously introduced in \cite{neuro-coding,neural-ring}.  A {\it coarse} RF code is obtained from a RF code by including all subsets of codewords, so that for each $\sigma \in 2^{[n]}$, 
$$\sigma \in \C  \text{\;\; if and only if \;\;} \bigcap_{i \in \sigma} U_i \neq \emptyset.$$
\end{definition}

Note that the codeword $\sigma = \{3,5,6\}$ in Figure 3 corresponds to stimuli in the shaded region, and {\it not} to the full intersection $U_3 \cap U_5 \cap U_6$.  Moreover, the subset $\tau  = \{3,5\} \subset \sigma$ is { not} an element of the RF code, since $U_3 \cap U_5 \subset U_6$.  Nevertheless, it often makes sense to also consider such subsets as codewords; for example, the cofiring of neurons 3 and 5 may still be observed, as neuron 6 may fail to fire even if the stimulus is in its receptive field.  
This is captured by the corresponding coarse RF code.

Coarse RF codes carry less detailed information about the underlying stimulus space~\cite{gap,neural-ring}, but turn out to be more ``natural'' in the context of symmetric threshold-linear networks because they have the structure of a simplicial complex.\footnote{In topology, this simplicial complex is called the {\it nerve} of the cover $\{U_1,\ldots,U_n\}$
(see \cite{BottTu, gap}).}  
This implies that coarse RF codes do not yield any Type 1 spurious states -- the spurious subsets -- when encoded in a network using the Encoding Rule.
Furthermore, both RF codes and coarse RF codes with {\it low-dimensional} receptive fields contain surprisingly few Type 2 spurious states -- the spurious cliques.  This follows from Helly's theorem, a classical theorem in convex geometry:
\medskip

\noindent{\bf Helly's theorem} \cite{Barvinok2002}.
Suppose that $U_1,...,U_k$ is a finite collection of convex subsets of $\RR^d$, for $d < k$.  If the intersection of any $d+1$ of these sets is nonempty, then the full intersection $\bigcap_{i=1}^k U_i$ is also nonempty.  
\medskip

\noindent To see the implications of Helly's theorem for RF codes, we define the notion of {\it Helly completion}.  

\begin{definition}
Let $\Delta$ be a simplicial complex on $n$ vertices, and let $\Delta_d = \{\sigma \in \Delta \mid |\sigma| \leq d+1\}$ denote its $d$-skeleton.  The {\em Helly completion}
 is the largest simplicial complex, $\bar{\Delta}_d,$ on $n$ vertices that has $\Delta_d$ as its $d$-skeleton.
\end{definition}
In other words, the Helly completion of a $d$-dimensional simplicial complex $\Delta_d$ is obtained by adding in all higher-dimensional faces in a way that is consistent with the existing lower-dimesional faces.
In particular, the Helly completion of any graph $G$ is the clique complex $X(G)$.  For a two-dimensional simplicial complex, $\Delta_2$, the Helly completion includes only cliques of the underlying graph $G(\Delta_2)$ that are consistent with $\Delta_2$.  For example, the Helly completion of the code in Section~\ref{sec:example} does {\it not} include the 3-cliques corresponding to empty (non-shaded) triangles in Figure 3A.
With this notion, Helly's theorem can now be reformulated as:

\begin{lemma}\label{lemma:helly}
Let $\C$ be a coarse $d$-dimensional RF code, corresponding to a set of place fields $\{U_1,...,U_n\}$ where each $U_i$ is a convex open set in $\RR^d$.
Then $\C$ is the Helly completion of its own $d$-skeleton: $\C  = {\bar \C}_d$.
\end{lemma}

This lemma indicates that low-dimensional RF codes -- whether coarse or not -- have a relatively small number of spurious cliques, 
since most cliques in $X(G(\C))$ are also in the Helly completion ${\bar \C}_d$ for small $d$.  In particular,
it implies that coarse RF codes of dimensions $d=1$ and $d=2$ are very natural codes for symmetric threshold-linear networks.

\begin{corollary}\label{cor:1dRF}
If $\C$ is a coarse one-dimensional RF code, then it is a clique complex: $\C = \bar \C_1 = X(G(\C))$.  Therefore, $\C$ can be exactly encoded 
using any universal $S$ in the Encoding Rule.  
\end{corollary}

\begin{corollary}
If $\C$ is a coarse two-dimensional RF code, then it is the Helly completion of its own $2$-skeleton, $\C = \bar \C_2$, which can be obtained 
from knowledge of all pairwise and triple intersections of receptive fields. 
\end{corollary}

\noindent For coarse two-dimensional RF codes, the only possible spurious cliques are therefore spurious triples and the larger cliques of $G(\C)$ that contain them. 
The spurious triples emerge when three receptive fields $U_i,U_j$ and $U_k$ have 
the property that each pair intersect, but $U_i \cap U_j \cap U_k = \emptyset$.   For generic arrangements 
of receptive fields, this is relatively rare, allowing these codes to be encoded nearly exactly using any universal $S$ in the Encoding Rule.  
In the next section, we illustrate this phenomenon in the case of two-dimensional place field codes.

\subsection{Encoding sparse place field codes in threshold-linear networks} \label{sec:PFcodes}

As seen in the last section, Helly's theorem sharply limits the number of spurious cliques that result from encoding
low-dimensional RF codes.  Here we illustrate this phenomenon explicitly in the case of sparse place field codes (PF codes).
In particular, we find that PF codes can be encoded nearly exactly from a very small, randomly selected sample of patterns.
The near-exact encoding of PF codes from highly undersampled data shows that they are ``natural'' codes for symmetric threshold-linear networks, as defined in Section~\ref{sec:spurious}.
\medskip

\noindent{\bf PF codes.} 
Let $\{U_1,...,U_n\}$ be a collection of convex open sets in $\RR^d$, where each $U_i$ is the {\it place field} corresponding to the $i$th neuron  \cite{O:1976,ON:1978}.  To such a set 
of place fields we associate a {\it $d$-dimensional PF code}, $\C$, defined as follows: 
for each $\sigma \in 2^{[n]}$, $\sigma \in \C$ if and only if the intersection $\bigcap_{i \in \sigma} U_i$ is nonempty.  

Note that in this definition, PF codes are coarse RF codes.
PF codes are experimentally observed in recordings of neural activity in rodent hippocampus \cite{PathIntegration}.  The elements of $\C$ correspond to subsets of neurons that may be co-activated as the animal's trajectory passes through a corresponding set of overlapping place fields.  Typically $d=1$ or $d=2$, corresponding to the standard ``linear track'' and ``open field'' environments \cite{Muller:1996}; recently, it has also been shown that flying bats possess $d=3$ place fields \cite{Ulanovsky:2013}.
\medskip

It is clear from Corollary~\ref{cor:1dRF} above that one-dimensional PF codes can be encoded exactly -- i.e., without any spurious states -- using any universal $S$ matrix in the Encoding Rule.  Two-dimensional PF codes have no Type 1 spurious states, but may have Type 2 spurious cliques.  For {\it sparse} PF codes, however, the spurious cliques can be further restricted (beyond what is expected from Helly's theorem) by choosing a $k$-sparse universal $S$.
  \medskip

\noindent{\bf Near-exact encoding of sparse PF codes.}
$\:$ Consider a two-dimensional PF code $\C$ that is {\it $k$-sparse}, so that no more than $k$ neurons can co-fire in a single pattern -- even if there are
higher-order overlaps of place fields.  Experimental evidence suggests that the fraction of active neurons is typically on the order of $5-10\%$ \cite{HP}, so we make
the conservative choice of $k = n/10$ (our results improve with smaller $k$).  In what follows, $S$ was chosen to
be $k$-sparse universal and $\varepsilon$ so that $0 < \varepsilon < \delta(S)$,
in order to control spurious cliques of size greater than $k$.  We also assume the worst-case-scenario of
$\P(W) = X_{k-1}(G(\C))$, providing an upper bound on the number of spurious cliques resulting from our Encoding Rule.
What fraction of the stored patterns are spurious?
This can be quantified by the following {\it error probability},
$$P_{\mathrm{error}} \od \dfrac{|\P(W) \setminus \C|}{|\P(W)|} = \dfrac{|X_{k-1}(G(\C))| - |\C|}{|X_{k-1}(G(\C))|},$$
which assumes all permitted sets are equally likely to be retrieved from among the stored patterns in $\P(W)$.
For exact encoding, $P_{\mathrm{error}} = 0$, while large numbers of spurious states will push $P_{\mathrm{error}}$ close to $1$.

\begin{figure}[h]  \label{fig:PFcode}
\begin{center}
\hspace{-.2in}
\includegraphics[width=6in]{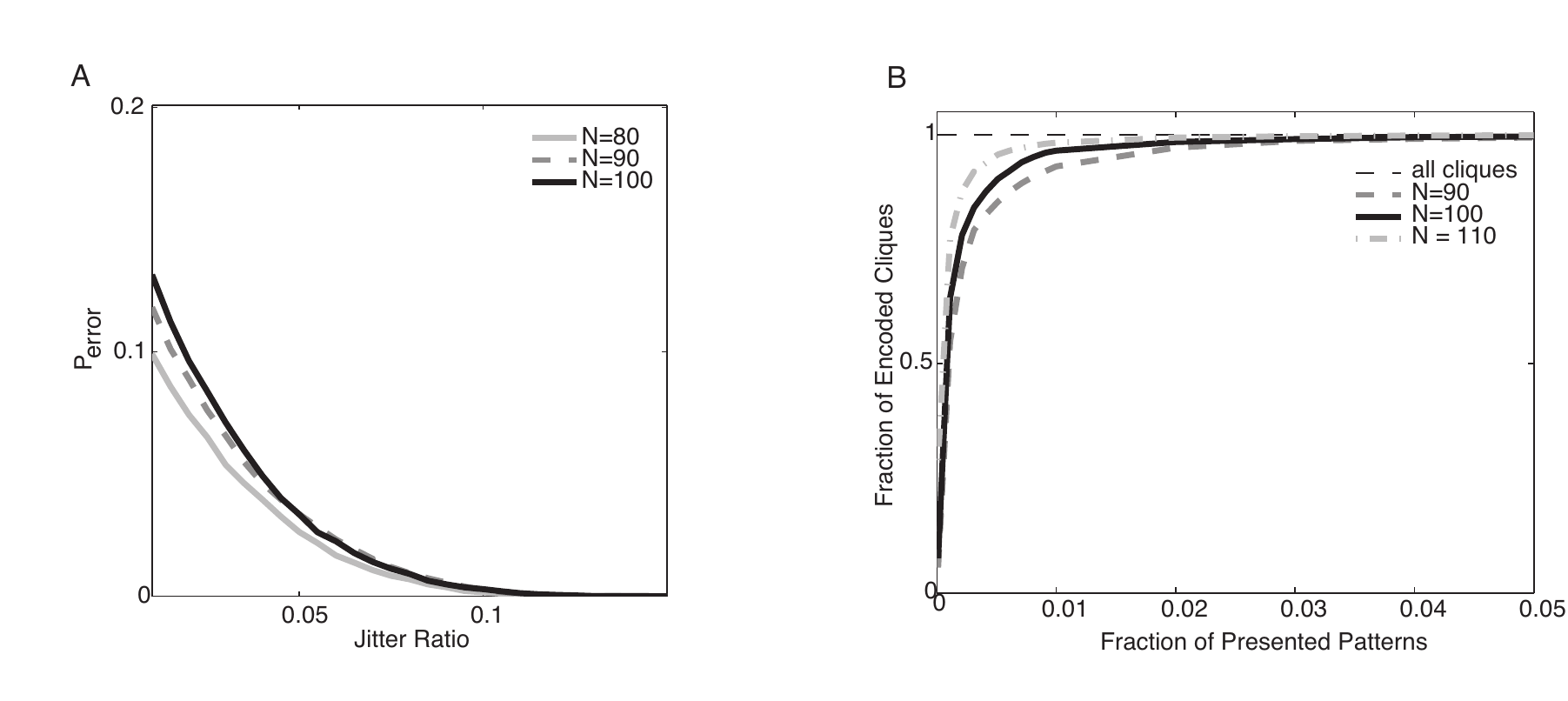}
\end{center}
\caption{\small PF encoding is near-exact, and can be achieved by presenting a small fraction of patterns.  (A) $P_{\mathrm{error}}$ was computed for randomly generated $k$-sparse PF codes
having $n = 80, 90$ and $100$ neurons and $k=n/10$.  For each jitter ratio, the average value of $P_{\mathrm{error}}$ over 100 codes is shown.
(B) For $n=90,100$ and $110$ neurons, $k$-sparse PF codes with jitter ratio $0.1$ were randomly generated and then randomly subsampled to contain a small fraction ($\leq 5\%$) of the total number of patterns.
After applying the Encoding Rule to the subsampled code, the number of encoded cliques was computed.   In each case, the fraction of encoded cliques for the subsampled code (as compared to the full PF code) was averaged over 10 codes. Cliques were counted using Cliquer \cite{Cliquer}, together with custom-made Matlab software.}
\end{figure}

To investigate how ``exactly''  two-dimensional PF codes are encoded, we generated random $k$-sparse PF codes with circular place fields, 
$n =$ 80-100 neurons, and $k = n/10$ (see Appendix E).   
Because experimentally observed place fields do not have precise boundaries, we also generated ``jittered'' codes, where spurious triples
were eliminated from the $2$-skeleton of the code if they did not survive after enlarging the place field radii from $r_0$ to $r_1$ by a {\em jitter ratio}, $(r_{1} - r_{0})/ r_{0}$.
This has the effect of eliminating spurious cliques that are unlikely to be observed in neural activity, as they correspond to very small regions in the underlying environment.
For each code and each jitter ratio (up to $\sim$ 0.1), we computed $P_{\mathrm{error}}$ using the formula above.   
Even without jitter, the error probability was small, and $P_{\mathrm{error}}$  decreased quickly to values near zero for 10\% jitter (Figure 4A).
\bigskip

\noindent{\bf Encoding full PF codes from highly undersampled sets of patterns.}
To investigate what fraction of patterns is needed to encode a two-dimensional  PF code using the Encoding Rule, we generated randomly subsampled codes from $k$-sparse PF codes.  
We then computed the number of patterns that would be encoded by a network if a 
subsampled code was presented.
Perhaps surprisingly, network codes obtained from highly subsampled PF codes (having only 1-5$\%$ of the patterns) 
are nearly identical to those obtained from full PF codes (Figure 4B).  This is because large numbers of ``spurious'' states emerge when encoding subsampled codes, 
but most correspond to patterns in the full code.  The spurious states of 
subsampled PF codes can therefore be {\it advantageous}, allowing networks to quickly encode full PF codes from only a small
fraction of the patterns.
\bigskip 

The results summarized in Figure 4 confirm the fact that sparse PF codes are natural codes, as they satisfy both properties P1 and P2 outlined in Section~\ref{sec:spurious}.
These codes can be encoded nearly exactly because they have very few spurious states.  The spurious cliques are limited by two factors: the implications of Helly's theorem (Section~\ref{sec:natural}) and their sparsity, enabling the choice of a $k$-sparse universal $S$ that automatically eliminates spurious cliques of size greater than $k$.

\section{Proofs} \label{sec:proofs}

To the best of our knowledge, all proofs in this section are original, as are the results presented in Theorems~\ref{thm:main-result},~\ref{thm:main-result2}, and~\ref{thm:perturbation}.  Theorem~\ref{thm:perturbation} is our core technical result, which we state and prove in Section~\ref{sec:thm4}.
It appears to be closely related to some relatively recent results in convex geometry, involving correlation matrices and the geometry of the ``elliptope'' \cite{MoniqueLaurent}.  
Our proof, however, relies only on classical distance geometry and well-known facts about stable symmetric matrices; these are summarized in Appendix A.
The key new insight that allows us to connect stability of matrices of the form $-11^T + \varepsilon A$ to Cayley-Menger determinants is Lemma~\ref{lemma:gendet}.  
In Section~\ref{sec:thm2-proof} we give the proofs of Proposition~\ref{prop:main-result2}, Theorem~\ref{thm:main-result} and Theorem~\ref{thm:main-result2}, which all rely on Theorem~\ref{thm:perturbation}.

\subsection{Statement of Theorem~\ref{thm:perturbation} and its proof}
\label{sec:thm4}

The statement of Theorem~\ref{thm:perturbation} uses the following definition and some new notation.

\begin{definition}
A {\em Hebbian} matrix $A$ is an $n \times n$ matrix satisfying $A_{ij}=A_{ji} \geq 0$ and $A_{ii} = 0$ for all $i,j \in [n]$. \end{definition}

\noindent The name reflects the fact that these are precisely the types of matrices that arise when synaptic weights are modified
by a Hebbian learning rule.  We also need the notation,
$$\RR_{\times}^n \od \{v \in \RR^n \mid v_i \neq 0 \;\text{for all}\; i \in [n]\}$$
for the set of vectors with all nonzero entries.  Note that for $v \in \RR^n_\times$, $- vv^T$ is a symmetric $n \times n$ rank 1 matrix with strictly negative diagonal. 
Next, given any $v \in \RR^n$ and any $n \times n$ matrix $A$, 
$$A^v \od \diag(v) A \diag(v)$$ 
denotes the matrix with entries $A^v_{ij} = v_i v_j A_{ij}$.  
We are now ready to state Theorem~\ref{thm:perturbation}.

\begin{theorem}\label{thm:perturbation}
Let $A$ be a Hebbian matrix and $\varepsilon > 0$.
For $v \in \RR^n_\times$, consider the perturbed matrix
$$M = -vv^T + \varepsilon A^v.$$ 
The following are equivalent:
\begin{enumerate}
\item $A$ is a nondegenerate square distance matrix.
\item There exists an $\varepsilon>0$ such that $M$ is stable.
\item There exists a $\delta > 0$ such that $M$ is stable for all $0 < \varepsilon < \delta$.
\item $0 < -\dfrac{\CM(A)}{\det A} < \infty$; and $M$ is stable if and only if $0<\varepsilon<-\dfrac{\CM(A)}{\det A}$.
\end{enumerate}
\end{theorem}
\medskip

The rest of this section is devoted to proving Theorem~\ref{thm:perturbation}.  A cornerstone of the proof  is the following lemma, which allows us to connect perturbations of rank 1 matrices to Cayley-Menger determinants. 

\begin{lemma}[determinant lemma]\label{lemma:gendet} Let $u,v \in \RR^n$.  For any real-valued $n \times n$ matrix $A$ and any $t \in \RR$,
$$\det( -uv^T + t \diag(u)A\diag(v)) = \det(\diag(u)\diag(v))\left(t^n\det A+t^{n-1}\CM(A)\right).$$
In particular, if $u = v \in \RR_\times^n$ and $t>0$, then
$$\sgn(\det( -vv^T + t A^v)) = \sgn\left(t\det A+\CM(A)\right),$$
where $\sgn: \RR \rightarrow \{\pm1,0\}$ is the sign function.  Moreover, taking $u = v = 1 \in \RR^n$ yields:
$$\det(-11^T + t A) = t^n\det A + t^{n-1}\CM(A).$$
\end{lemma}

\begin{proof}[Proof of Lemma~\ref{lemma:gendet}]
Note that for any $n \times n$ matrix $A$, $t \in \RR$, and $u,v \in \RR^n$, we have
$$\det(-uv^T+t\diag(u)A\diag(v)) = \det(\diag(u)\diag(v))\det(-11^T + tA),$$
where ${-11^T}$ is as usual the rank 1 matrix of all $-1$s.  It thus suffices to show that
$$\det(-{11^T} + tA) = t^n \det A + t^{n-1}\CM(A),$$
where $\CM(A)$ is the Cayley-Menger determinant of $A$.

Let $w, z \in \IR^n$, and let $Q$ be any $n \times n$ matrix. We have
$$
  \det \begin{bmatrix} 1 & z^T \\ w & Q \end{bmatrix} = \det(Q  - w z^T),
$$
where we have used the well-known formula for computing the determinant
of a $2 \times 2$ block matrix.\footnote{The formula $ \det \begin{bmatrix} A & B \\ C & D\end{bmatrix} = \det(A)\det(D-C A^{-1}B)$ applies so long as $A$ is invertible.  
It follows from observing that
$\left[\begin{array}{cc} I & 0 \\ -CA^{-1} & I \end{array}\right]
\left[\begin{array}{cc} A & B \\ C & D \end{array}\right] = 
\left[\begin{array}{cc} A & B \\ 0 & -CA^{-1}B+D \end{array}\right].$}
On the other hand, the usual cofactor expansion along the first row gives 
$$
 \det \begin{bmatrix} 1 & z^T \\ w & Q \end{bmatrix} 
  = \det(Q) + \det\begin{bmatrix} 0 & z^T \\ w & Q \end{bmatrix}. 
$$
Therefore,
$$
 \det(- wz^T+ Q) = \det(Q) + \det\begin{bmatrix} 0 & z^T \\ w & Q \end{bmatrix} .
$$
In particular, taking $w = z = 1 \in \RR^n$ (the column vector of all ones) and $Q = tA$, we have
$\det(-{11^T} + tA) = \det(tA) + \CM(tA) = t^n \det A + t^{n-1} \CM(A).$
\end{proof}

Finally, to prove Theorem~\ref{thm:perturbation} we will need the following technical lemma.

\begin{lemma}\label{lemma:D0}
Fix $v \in \RR_\times^n$, and let $A$ be an $n \times n$ Hebbian matrix. If $(-1)^n\CM(A) \leq 0$, then $-vv^T+tA^v$ is not stable for any $t >0$.  
In particular, if there exists a $t>0$ such that $-vv^T + tA^v$ is stable, then $(-1)^n\CM(A)>0$.
\end{lemma}

\noindent For its proof, we will need a simple convexity result.

\begin{lemma}\label{lemma:convexity}
Let $M,N$ be real symmetric $n \times n$ matrices so that $M$ is negative semidefinite (i.e., all eigenvalues are $\leq 0$) and $N$ is strictly negative definite (i.e., stable, with all eigenvalues $<0$).  Then $tM + (1-t)N$ is strictly negative definite (i.e., stable) for all $0\leq t<1.$ 
\end{lemma}

\begin{proof}
$M$ and $N$ satisfy $x^TMx \leq 0$ and $x^TNx < 0$ for all $x \in \RR^n$,  so we have $x^T(tM+(1-t)N)x<0$ for all nonzero $x \in \RR^n$ if $0 \leq t < 1$.
\end{proof}

The proof of Lemma~\ref{lemma:D0} relies on Lemmas~\ref{lemma:gendet} and~\ref{lemma:convexity}, which we have just proven,
and also on some well-known results from classical distance geometry that are
collected in Appendix A.  These include facts about stable symmetric
matrices (Cauchy's interlacing theorem, Corollary~\ref{cor:simplicialprop}, and Lemma~\ref{lemma:signcondition}) as well
as facts about square distance matrices (Lemma~\ref{lemma:realizability}, Proposition~\ref{prop:sd-det}, and
Corollary~\ref{cor:nondeg}).  These facts are also used in the proof of Theorem~\ref{thm:perturbation}.

\begin{proof}[{Proof of Lemma~\ref{lemma:D0}}]
Since A is symmetric, so are $A^v$ and $-vv^T+tA^v$ for any $t$.  Hence, if any principal submatrix of $-vv^T+tA^v$ is unstable then $-vv^T+tA^v$ is also unstable, by Corollary~\ref{cor:simplicialprop}.  Therefore, without loss of generality we can assume $(-1)^{|\sigma|}\CM(A_\sigma) > 0$ for all proper principal submatrices $A_\sigma$, with $|\sigma|<n$ (otherwise, we use this argument on a smallest principal submatrix such that $(-1)^{|\sigma|}\CM(A_\sigma) 
\leq 0$).  By Lemma~\ref{lemma:realizability}, this implies that $A_\sigma$ is a nondegenerate square distance matrix for all $\sigma$ such that $|\sigma|<n$, and so we know by Proposition~\ref{prop:sd-det} that $(-1)^{|\sigma|}\det A_\sigma < 0$  
and that each $A_\sigma$ such that $1<|\sigma|<n$ has one positive eigenvalue and all other eigenvalues negative.

We prove the lemma by contradiction.  Suppose there exists a $t_0>0$ such that $-vv^T + t_0A^v$ is stable.  Applying Lemma~\ref{lemma:convexity} with $M = -vv^T$ and $N = -vv^T+t_0A^v$, we have that $-vv^T+(1-t)t_0A^v$ is stable for all $0 \leq t < 1$.  It follows that $-vv^T + tA^v$ is stable for all $0<t\leq t_0$. 
 Now Lemma~\ref{lemma:signcondition} 
 implies that $(-1)^n \det(-vv^T + t A^v) >0$ for all $0<t\leq t_0$. By
  Lemma~\ref{lemma:gendet}, this is equivalent to having $(-1)^n(t\det
  A + \CM(A)) > 0$ for all $0<t \leq t_0$.
 By assumption, $(-1)^n \CM(A) \leq 0$. But, if $(-1)^n\CM(A) < 0$, then there would exist a small enough $t>0$ such that $(-1)^n(t\det  A + \CM(A)) < 0$.  Therefore we conclude that $\CM(A) = 0$ and hence $(-1)^n\det A >0$.

Next, let $\lambda_1 \leq ... \leq \lambda_n \leq \lambda_{n+1}$ denote the eigenvalues of the Cayley-Menger matrix $ CM(A) = \left[\begin{array}{cc} 0 & 1^T \\ 1 &A\end{array}\right]$, and observe that $A$, $A_{[n-1]}$, and $CM(A_{[n-1]})$ are all principal submatrices of $CM(A)$.  Since everything is symmetric, Cauchy's interlacing theorem applies.  We have seen above that $A_{[n-1]} $ has one positive eigenvalue and all others negative, so by Cauchy interlacing $\lambda_{n+1}>0$ and $\lambda_{n-2}<0$.  Because $\CM(A) = \det CM(A) = 0$, then $CM(A)$ must have a zero eigenvalue, while $\det A \neq 0$ implies that it is unique.   We thus have two cases.

Case 1: Suppose $\lambda_{n-1} = 0$ and thus $\lambda_n >0$.
Since we assume $(-1)^{n-1}\CM(A_{[n-1]}) >0$, the $n \times n$ matrix $CM(A_{[n-1]})$ must have an odd number of positive eigenvalues, but 
by Cauchy interlacing the top two eigenvalues must be positive, so we have a contradiction.

Case 2: Suppose $\lambda_{n} = 0$ and thus $\lambda_{n-1}<0$.  Then by Cauchy interlacing $A$ has exactly one positive eigenvalue.  On the other hand, the fact that $(-1)^n \det A > 0$ implies that $A$ has an even number of positive eigenvalues, which is a contradiction.
\end{proof}

\noindent We can now prove Theorem~\ref{thm:perturbation}.

\begin{proof}[{Proof of Theorem~\ref{thm:perturbation}}]
We prove $(4) \Rightarrow (3) \Rightarrow (2) \Rightarrow (1) \Rightarrow (4)$.  

(4) $\Rightarrow$ (3) $\Rightarrow$ (2) is obvious.

(2) $\Rightarrow$ (1): Suppose there exists a $t>0$ such that $-vv^T+tA^v$ is stable.  Then, by 
Corollary~\ref{cor:simplicialprop} and Lemma~\ref{lemma:D0},  $(-1)^{|\sigma|}\CM(A_\sigma) > 0$ for all principal submatrices $A_\sigma$.  By  Lemma~\ref{lemma:realizability}  it follows that $A$ is a nondegenerate square distance matrix.

(1) $\Rightarrow$ (4): 
Suppose $A$ is a nondegenerate square distance matrix.  
By Lemma~\ref{lemma:realizability} we have $(-1)^{|\sigma|}\CM(A_\sigma)>0$ for all $A_\sigma$, while Proposition~\ref{prop:sd-det} implies
$(-1)^{|\sigma|}\det(A_\sigma)<0$ for all $A_\sigma$ with $|\sigma|>1$.   This implies that
for $|\sigma|>1$  we have $-\dfrac{\CM(A_\sigma)}{\det(A_\sigma)}>0$ (by Corollary~\ref{cor:nondeg}),  and that if $\varepsilon>0$,
$$(-1)^{|\sigma|}\left(\varepsilon \det(A_\sigma) + \CM(A_\sigma)\right) > 0
\;\; \Leftrightarrow \;\; \varepsilon<-\dfrac{\CM(A_\sigma)}{\det(A_\sigma)}.$$ 
Applying now  Lemma~\ref{lemma:gendet}, 
 $$(-1)^{|\sigma|}\det(-vv^T+\varepsilon A^v)_\sigma>0
 \;\; \Leftrightarrow \;\; \varepsilon<-\dfrac{\CM(A_\sigma)}{\det(A_\sigma)}.$$
 For $|\sigma| = 1$, we have diagonal entries $A_\sigma = A_\sigma^v = 0$ and $(-vv^T)_\sigma < 0$, so $(-1)\det(-vv^T + \varepsilon A^v)_\sigma > 0$ for all $\varepsilon$.
Using Lemma~\ref{lemma:signcondition}, we conclude (assuming $\varepsilon>0$):
$$-vv^T + \varepsilon A^v \mbox{ is stable}\;\; \Leftrightarrow \;\;\varepsilon < \delta,$$ where 
$$\delta = \min \left\{-\dfrac{\CM(A_\sigma)}{\det(A_\sigma)}\right\}_{\sigma \subseteq [n]}>0.$$
It remains only to show that  $\delta = -\CM(A)/\det(A).$  Note that we can not use Lemma~\ref{lemma:inequality} from the Main Text because that lemma follows from Proposition~\ref{prop:main-result2}, and is hence a consequence of Theorem~\ref{thm:perturbation}.  

On the other hand, because the matrix $-vv^T+\varepsilon A^v$ changes from stable to unstable at $\varepsilon = \delta$, by continuity of the eigenvalues as functions of $\varepsilon$ it must be that 
$$\det(-vv^T + \delta A^v) = 0.$$
Using Lemma~\ref{lemma:gendet} it follows that
$\delta \det(A) + \CM(A) = 0,$
which implies $\delta = -\CM(A)/\det(A).$
\end{proof}

\subsection{Proofs of Proposition~\ref{prop:main-result2}, Theorem~\ref{thm:main-result} and Theorem~\ref{thm:main-result2}}\label{sec:thm2-proof}

Here we prove our main results from Sections~\ref{sec:main-result} and~\ref{sec:exact}.  We begin with the proof of Proposition~\ref{prop:main-result2}.

\begin{proof}[Proof of Proposition~\ref{prop:main-result2}]
Setting $v = 1 \in \RR_\times^n$ (the column vector of all ones) in Theorem~\ref{thm:perturbation} yields a slightly weaker version of Proposition~\ref{prop:main-result2}, as the hypothesis in Theorem~\ref{thm:perturbation} is that $A$ is {\it Hebbian}, which is more constrained than the Proposition~\ref{prop:main-result2} hypothesis that $A$ is symmetric with zero diagonal.  
To see why Proposition~\ref{prop:main-result2} holds more generally, suppose $A$ is symmetric with zero diagonal but {\it not} Hebbian.  Then there exists an off-diagonal pair of negative entries, $A_{ij} = A_{ji} < 0$, and the $2 \times 2$ principal submatrix 
$$(-11^T + \varepsilon A)_{\{ij\}} = \left(\begin{array}{cc} -1 & -1 + \varepsilon A_{ij} \\ -1 + \varepsilon A_{ij} & -1\end{array}\right)$$
is unstable as it has negative trace and negative determinant.  It follows from Cauchy's interlacing theorem (see Corollary~\ref{cor:simplicialprop} in Appendix A) that $-11^T + \varepsilon A$ is unstable for any $\varepsilon > 0$.  Correspondingly, condition (a) in Proposition~\ref{prop:main-result2} is violated, as the existence of negative entries guarantees that $A$ cannot be a nondegenerate square distance matrix. 
\end{proof}

To prove Theorems~\ref{thm:main-result} and~\ref{thm:main-result2}, we will need the following two corollaries of Proposition~\ref{prop:main-result2}.  
First, recall the definitions for $\geom(A), \geom_\varepsilon(A)$, $\delta(A)$ from Section~\ref{sec:main-result}.
Applying Proposition~\ref{prop:main-result2} to each of the 
principal submatrices of the perturbed matrix $-11^T + \varepsilon A$ we obtain:

\begin{corollary}\label{cor:cor4}  If $A$ is a symmetric matrix with zero diagonal, and $\varepsilon>0$, then
$$\stab(-11^T + \varepsilon A) = \geom_\varepsilon(A).$$
For $0<\varepsilon < \delta(A)$, $\stab(-11^T + \varepsilon A) = \geom(A).$
\end{corollary}

Next, recall that $X(G)$ is the clique complex of the graph $G$.

\begin{corollary} \label{cor:cor5} Let $A$ be a symmetric $n \times n$ matrix with zero diagonal, and $\varepsilon > 0$.  Let $G$ be the graph on $n$ vertices having $(ij) \in G$ if and only if $A_{ij} \geq 0$.  For any $n \times n$ matrix $S$ with $S_{ij}=S_{ji} \geq 0$ and $S_{ii} = 0$, if $S$ ``matches'' $A$ on $G$ (i.e., if $S_{ij} = A_{ij}$ for all $(ij) \in G$), then
$$\geom_\varepsilon(A) = \geom_\varepsilon(S) \cap X(G).$$
In particular, $\geom(A) = \geom(S) \cap X(G).$
\end{corollary}

We can now prove Theorems~\ref{thm:main-result} and~\ref{thm:main-result2}.

\begin{proof}[Proof of Theorem~\ref{thm:main-result}]
Any network $W$ obtained via the Encoding Rule (equation~\eqref{eq:encoding-rule}) has the form $-D+W = -11^T + \varepsilon A$, where $A$ is symmetric with zero diagonal and ``matches'' the (nonnegative) synaptic strength matrix $S$ precisely on the entries $A_{ij}$ such that $(ij) \in G(\C)$.  All other off-diagonal entries of $A$ are negative.  It follows that
\begin{eqnarray*}
\hspace{.35in}\P(W) &=& \stab(-11^T + \varepsilon A) = \geom_\varepsilon(A)\\
 &=& \geom_\varepsilon(S) \cap X(G(\C)),
\end{eqnarray*}
where the last two equalities are due to Corollaries~\ref{cor:cor4} and~\ref{cor:cor5}, respectively.
\end{proof}

\begin{proof}[Proof of Theorem~\ref{thm:main-result2}]
$(\Leftarrow)$  This is an immediate consequence of Theorem~\ref{thm:main-result}.

$(\Rightarrow)$  Suppose there exists a symmetric network $W$ with $\P(W) = \C$, and observe by Theorem~\ref{thm:thm1} that  
$\P(W) = \stab(-11^T+A),$
for some symmetric $n \times n$ matrix $A$ with zero diagonal.  By Corollaries~\ref{cor:cor4} and~\ref{cor:cor5}, 
$$\C = \P(W) = \geom_\varepsilon(A) = \geom_\varepsilon(S) \cap X(G),$$
where $\varepsilon = 1$, $G$ is the graph associated to $A$ (as in Corollary~\ref{cor:cor5}) and $S$ is an $n \times n$ matrix with $S_{ij} = S_{ji} \geq 0$ and zero diagonal that ``matches'' $A$ on $G$.  It remains only to show that $\geom_\varepsilon(S) \cap X(G) = \geom_\varepsilon(S) \cap X(G(\C))$. 
Since $\C = \geom_\varepsilon(A)$, any element $\{ij\} \in \C$ must have corresponding $A_{ij}>0$, so $G(\C) \subseteq G$ and hence $X(G(\C)) \subseteq X(G)$.
On the other hand, $\C = \C \cap X(G(\C))$, so we conclude that $\C =  \geom_\varepsilon(S) \cap X(G(\C)).$ 
\end{proof}

\section{Discussion} \label{sec:discussion}
 
 Understanding the relationship between the connectivity matrix and the  
 activity patterns of a neural network is one of the central challenges in theoretical 
 neuroscience.  We have found that in the context of symmetric threshold-linear networks, 
 one can obtain an unexpectedly precise understanding of
 the binary activity patterns stored by network steady states.  
 In particular, we have arrived at a complete and precise combinatorial characterization of spurious states, 
 something that has not yet been achieved in the context of the Hopfield model 
 \cite{amit_spin-glass_1985, AGS87, Amit1989, HertzKroghPalmer, RoudiTreves2003}.
 Moreover, we have shown that network solutions
 to the NE problem can be obtained constructively, using a simple Encoding Rule.
  A new concept that
 emerges from our results is that of {\it geometric balance}, whereby the excitatory synapses between neurons
 in a stored pattern must satisfy a set of geometric constraints, ensuring they are appropriately bounded
 and balanced in their strengths.  
 
As a consequence of our main results, we have discovered that threshold-linear networks naturally
 encode neural codes arising from low-dimensional receptive fields (such as place fields)
 while introducing very few spurious states.  Remarkably, these codes can be ``learned'' by 
 the network from a highly undersampled set of patterns.
 Neural codes representing (continuous) parametric stimuli, such as place field codes, have typically
 been modeled as arising from continuous attractor networks whose synaptic matrices
 have symmetric ``Mexican hat''-type connectivity \cite{Ben-Yishai1995,PathIntegration}.  
This is in large part due to the fact that there
 is a well-developed mathematical handle on these networks 
\cite{Amari1977,Bressloff12,IHT11}.
Our work shows that one can have fine mathematical control over a much wider class of networks, encompassing all symmetric connectivity matrices.  
It may thus provide a novel foundation for understanding -- and engineering -- neural networks with prescribed steady state properties. 
 
\subsection*{Acknowledgments}
The authors would like to thank Christopher Hillar, Caroline Klivans, Katie Morrison and Bernd Sturmfels for valuable comments, 
as well as Nora Youngs and Zachary Roth for assistance with figures.  
CC was supported by NSF DMS 0920845, NSF DMS 1225666, 
a Woodrow Wilson Fellowship, and the Sloan Research Fellowship.  
AD was supported by the Max Planck Society and the DFG via SFB/Transregio 71 ``Geometric Partial Differential Equations.''
VI was supported by NSF DMS 0967377 and NSF DMS 1122519.

\section{Appendices} \label{sec:appendices}

\subsection{Appendix A: Stable symmetric matrices and square-distance matrices}

In this appendix we review some classical facts about stable symmetric matrices and square-distance matrices that are critical to many of our proofs.  Everything in this section is well-known.

\subsubsection{Stable symmetric matrices}
Here we summarize some well-known facts about the stability of symmetric matrices that we use in various proofs.
The first is Cauchy's interlacing theorem, which relates eigenvalues of a symmetric matrix to those of its principal submatrices.
Recall that the eigenvalues of a symmetric matrix are always real.

\begin{theorem}[Cauchy's interlacing theorem \cite{HornJohnson}]\label{thm:cauchy} Let $A$ be a symmetric $n \times n$ matrix, and let $B$ be an $m \times m$ principal submatrix of $A$.  If the eigenvalues of $A$ are $\alpha_1 \leq ...\alpha_j... \leq \alpha_n$ and those of $B$ are $\beta_1 \leq ... \beta_j ... \leq \beta_m$, then 
$\alpha_j \leq \beta_j \leq \alpha_{n-m+j}$ for all $j$.
\end{theorem}

\noindent Some immediate consequences of this theorem are:

\begin{corollary}\label{cor:simplicialprop}
Any principal submatrix of a stable symmetric matrix is stable.  Any symmetric matrix containing an unstable principal submatrix is unstable.
\end{corollary}

\begin{corollary}\label{cor:stab-simplicial}
Let $A$ be a symmetric $n \times n$ matrix with strictly negative diagonal.  Then $\stab(A)$ is a simplicial complex.
\end{corollary}

\begin{proof}  First, recall the definitions of $\stab(A)$ and simplicial complex from Section~\ref{sec:permitted}.
We need to check the two properties in the definition of a simplicial complex.
Property (1) holds for $\stab(A)$, because $A$ has strictly negative diagonal.
Property (2) follows from Corollary~\ref{cor:simplicialprop}.
\end{proof}

\noindent Another well-known consequence of Cauchy's interlacing theorem is the following lemma. Here $A_{[k]}$ refers to the principal submatrix obtained by taking the upper left $k \times k$ entries of $A$.

\begin{lemma}\label{lemma:signcondition}  Let $A$ be a real symmetric $n \times n$ matrix.  Then the following are equivalent:
\begin{enumerate}
\item $A$ is a stable matrix.
\item $(-1)^k\det(A_{[k]})>0$ for all $1 \leq k \leq n$.
\item $(-1)^{|\sigma|}\det(A_\sigma)>0$ for every $\sigma \subseteq [n]$.
\end{enumerate}
\end{lemma}

\begin{proof}  We prove (1) $\Leftrightarrow$ (2).  The equivalence between (1) and (3) follows using a very similar argument.

$(\Rightarrow)$ Assume $A$ is stable.  Then $\lambda_1(A) \leq ... \leq \lambda_n(A) < 0.$  By Cauchy's interlacing theorem, $\lambda_1(A) \leq \lambda_i(A_{[k]}) \leq \lambda_n(A)$ for all $i=1,...,k$ and $k=1,...,n$.  Therefore, all eigenvalues of the matrices $A_{[k]}$ are strictly negative, and hence $(-1)^k\det(A_{[k]})>0$ for all $k$.

$(\Leftarrow)$ We prove this by induction.  The base case is $n=1$: indeed, a $1 \times 1$ matrix $A = [a]$ is stable if $-\det(A_{[1]})>0$, i.e. if $a<0$. Now suppose $(\Leftarrow)$ of the theorem is true for $(n-1) \times (n-1)$ matrices, and also that $(-1)^k\det(A_{[k]}) > 0$, $k=1,...,n$, for an $n \times n$ matrix $A$.  This implies $A_{[n-1]}$ is stable.  By Cauchy's interlacing theorem, the highest eigenvalue $\lambda_{n-1}(A_{[n-1]})$ lies between the top two eigenvalues of $A$:
$$\lambda_{n-1}(A) \leq \lambda_{n-1}(A_{[n-1]}) \leq \lambda_n(A).$$
The stability of $A_{[n-1]}$ thus implies $\lambda_1(A) \leq ... \leq \lambda_{n-1}(A) < 0$.  It remains only to check that $\lambda_{n}(A)<0$.  For $n$ even, $(-1)^n\det(A_{[n]})>0$ implies $\lambda_1(A)\lambda_2(A)\cdots\lambda_n(A)>0$, hence $\lambda_n(A)<0$.  For $n$ odd, $(-1)^n\det(A_{[n]})>0$ implies $\lambda_1(A)\lambda_2(A)\cdots\lambda_n(A)<0$, hence $\lambda_n(A)<0$.  It follows that $A$ is stable.
\end{proof}

\subsubsection{Square distance matrices}
Recall from Section~\ref{sec:main-result} the definitions of square distance matrix, nondegenerate square distance matrix, and Cayley-Menger determinant.
Our convention is that the $1 \times 1$ zero matrix $[0]$ is a nondegenerate square distance matrix,
as $|\CM([0])| = 1 > 0$.  As an example, a $3 \times 3$ symmetric matrix $A$ with zero diagonal is a  nondegenerate square distance matrix if and only if the off-diagonal entries $A_{ij}$ are all positive, and their square roots ($\sqrt{A_{12}}, \sqrt{A_{13}},$ and $\sqrt{A_{23}}$) satisfy all three triangle inequalities.

There are two classical characterizations of square distance matrices.  The first, due to Menger \cite{Blumenthal}, relies on Cayley-Menger determinants.  The second, due to Schoenberg \cite{Schoenberg}, uses eigenvalues of principal submatrices.  Both are needed for our proof of Theorem~\ref{thm:perturbation}.

The relationship between Cayley-Menger determinants and simplex volumes is well-known:

\begin{lemma}\label{lemma:cayley-menger}
 Let $p_1,..,p_k$ be $k$ points in a Euclidean space. Assume that $A_{ij}=\norm{p_i-p_j}^2$ is the matrix of square distances between these points. 
 Then the $(k-1)$-dim volume $V$ of the convex hull of the points $\{p_i\}_{i=1}^k$ can be computed as 
\begin{equation}\label{k-volume}
\hspace{.75in} V^2=\frac{(-1)^{k}}{2^{(k-1)}\left(\left(k-1\right)!\right)^2}\CM(A).
\end{equation}
In particular, if $A$ is a degenerate square distance matrix then $\CM(A)=0$.
\end{lemma}

\noindent This leads to Menger's characterization of square distance matrices.  Recall that $A_\sigma$ is the principal submatrix obtained by restricting $A$ to the index set $\sigma$.

\begin{lemma}\label{lemma:realizability}
Let $A$ be an $n \times n$ matrix satisfying $A_{ij}=A_{ji} \geq 0$ and $A_{ii} = 0$ for all $i,j \in [n]$ (i.e., $A$ is a Hebbian matrix).  Then, 
\begin{enumerate}
\item $A$ is a square distance matrix if and only if 
 $(-1)^{|\sigma|} \CM(A_\sigma) \geq 0$ for every $A_\sigma$.
\item $A$ is a nondegenerate square distance matrix if and only if $(-1)^{|\sigma|} \CM(A_\sigma) >0$ for every  $A_\sigma$. 
\end{enumerate}
\end{lemma}
\begin{proof}
(1) is equivalent to the Corollary of Theorem 42.2 in \cite{Blumenthal}.
(2) is equivalent to Theorem 41.1 in \cite{Blumenthal}. 
\end{proof}

Schoenberg's characterization implies that if a matrix
is a square distance matrix, then the determinant of any principal submatrix has opposite sign to that of its Cayley-Menger determinant.

\begin{proposition}\label{prop:sd-det}
Let $A$ be an $n \times n$ square distance matrix that is not the zero matrix.  Then: 
\begin{enumerate}
\item $A$ has one strictly positive eigenvalue and $n-1$ eigenvalues that are less than or equal to zero. In particular,
  $(-1)^{|\sigma|} \det(A_\sigma) \leq 0$ for
  every principal submatrix $A_{\sigma}$.
\item If $A$ is a nondegenerate square distance matrix, $A$ has no zero eigenvalues and\\
  $(-1)^{|\sigma|} \det(A_\sigma) < 0$ for every principal submatrix $A_\sigma$ with $|\sigma|>1$. 
  \end{enumerate}
\end{proposition}

\begin{proof}
This Proposition is contained in \cite[Theorem 6.2.16]{MoniqueLaurent}.
It can also be proven directly from Theorem 1 of Schoenberg's 1935 paper \cite{Schoenberg}. 
 \end{proof}
 
 \begin{corollary}\label{cor:nondeg}
If $A$ is an $n \times n$ nondegenerate square distance matrix with $n > 1$, then
$$-\dfrac{\CM(A)}{\det A} > 0.$$
\end{corollary}

\subsection{Appendix B: Some more facts about permitted sets of symmetric threshold-linear networks}
 \label{sec:nonsym-rank1}

This section proves a few additional (and novel) facts about permitted sets in symmetric threshold-linear networks.  These were not included in the main text in order not to disrupt the flow of the exposition.

Recall from Theorem~\ref{thm:stable-set} that
the permitted sets $\P(W)$, where $W$ is a threshold-linear network with dynamics given by equation~\eqref{eq:dynamics}, always have the form
$$\P(W) = \stab(-D+W).$$
Here we show that when $W$ is {\it symmetric} (like the networks obtained using the Encoding Rule~\eqref{eq:encoding-rule}), $\P(W)$ can always be expressed as $\stab(-11^T+A)$ or $\stab(-xy^T+B)$, where $-xy^T$ is any rank $1$ matrix
having strictly negative diagonal, and $A, B$ are square matrices with zero diagonal.  
In what follows, we use the notation $\RR_{\times}^n$ and $A^v$ defined at the beginning of Section~\ref{sec:thm4}.

\begin{lemma}\label{lemma:DMD}
Let $M$ be a symmetric $n \times n$ matrix, and $v \in \RR^n_\times$.  Then,
$$\stab(M^v) = \stab(M).$$
In other words, a principal submatrix $M_\sigma^v$ is stable if and only if $M_\sigma$ is stable.
\end{lemma}

\begin{proof}  
By Lemma~\ref{lemma:signcondition}, $\tau \in \stab(M)$ if and only if 
$(-1)^{|\sigma|} \det(M_\sigma) > 0$
for every $\sigma \subseteq \tau$.  Observe that, since $M^v = \diag(v)M\diag(v)$, we have $\sgn(\det(M^v_\sigma)) = \sgn(\det(M_\sigma))$ for all $\sigma \subseteq [n]$.  It follows that $\tau \in \stab(M^v)$ if and only if $\tau \in \stab(M)$.  
\end{proof}

\begin{lemma}\label{lemma:13}
For any symmetric threshold-linear network $W$ on $n$ neurons, there exists a symmetric $n \times n$ matrix $A$ with zero diagonal such that
$$\P(W) = \stab(-11^T + A).$$
\end{lemma}

\begin{proof}
Let $x  \in \RR_{\times}^n$ be the vector such that $\diag(-xx^T) = \diag(-D+W)$, and write
$$-D+W = -xx^T + (-D+W + xx^T),$$
where the term in parentheses is symmetric and has zero diagonal.  This can be rewritten as
$$-D+W = \diag(x)(-11^T + A)\diag(x) = (-11^T + A)^x,$$
where 
$$A = \diag(x)^{-1} (-D+W + xx^T) \diag(x)^{-1}$$
is a symmetric $n \times n$ matrix with zero diagonal.   It follows from Lemma~\ref{lemma:DMD} that
$\P(W) = \stab(-D+W) = \stab(-11^T+A).$ 
\end{proof}

Lemma~\ref{lemma:13} implies that all sets of permitted sets  $\P(W)$ for symmetric networks $W$ have the form
$\P(W) = \stab(-11^T + A)$,
where $A$ is a symmetric matrix having zero diagonal.  The following Proposition implies
that all such codes can also be obtained by perturbing around {\it any} rank 1 matrix with negative diagonal, not necessarily
symmetric.  
Note that if $x,y \in \RR^n_\times$, the rank 1 matrix $-xy^T$ has strictly negative diagonal
if and only if $x_iy_i>0$ for all $i \in [n]$.

\begin{proposition}\label{prop:-1perturbation}
Fix  $x,y \in \RR^n_\times$ with $x_iy_i>0$ for all $i \in [n]$.  For any symmetric threshold-linear network $W$
on $n$ neurons, there exists an $n \times n$ matrix $B$ with zero diagonal such that
$$\P(W)  = \stab(-xy^T + B) .$$
\end{proposition}
\noindent The proof of this Proposition constructs the matrix $B$ explicitly, and relies on the following Lemma.

\begin{lemma}\label{lemma:conj}
Let $M$ be any $n\times n$ matrix, and $T$ an  $n \times n$ invertible diagonal matrix. Then $$\stab(TMT^{-1}) = \stab(M).$$
\end{lemma}
\begin{proof}
We have $(TMT^{-1})_{\sigma} = T_{\sigma}M_{\sigma}T^{-1}_{\sigma}$. Since conjugation preserves
the eigenvalue spectrum, the statement follows.  
\end{proof}

\begin{proof}[Proof of Proposition~\ref{prop:-1perturbation}]
Let $W$ be a symmetric threshold-linear network on $n$ neurons.  By Lemma~\ref{lemma:13}, there exists a symmetric $n \times n$ matrix $A$ with zero diagonal such that $\P(W) = \stab(-11^T+A)$.
It thus remains only to construct an $n \times n$ matrix $B$ with zero diagonal such that 
$$\stab(-xy^T + B) = \stab(-11^T + A).$$
We prove that this can always be done in two steps: first, we prove that it can be done
in the special case $x = y$, and then we show that $B$ can be constructed in general.

Step 1:  Fix $x = y \in \RR_\times^n$, and observe that
 $-xx^T+A^x = (-11^T+A)^x$, so
by Lemma~\ref{lemma:DMD} we have $\stab(-xx^T+A^x) = \stab(-11^T+A)$.
Letting $B = A^x$ we obtain the desired statement.

Step 2: Fix $x, y \in \RR_\times^n$ so that $x_iy_i>0$ for all $i \in [n]$, and let $T$ be the diagonal matrix
with entries $T_{ii} = \sqrt{y_i/x_i}.$
Then
$$( T (-xy^T) T^{-1})_{ij} = \sqrt{\dfrac{y_i}{x_i}}(-x_iy_j)\sqrt{\dfrac{x_j}{y_j}}=-\sqrt{x_iy_i}\sqrt{x_jy_j},$$
so $T(-xy^T)T^{-1} = -zz^T$ for $z \in \RR_\times^n$ having entries $z_i = \sqrt{x_iy_i}$.
It follows from Step 1 that
$\stab(-11^T+A) = \stab(-zz^T+A^z) = \stab(T(-xy^T)T^{-1} + A^z).$
Let $$B = T^{-1}A^zT.$$  Then, using Lemma~\ref{lemma:conj},
$\stab(-xy^T+B) = \stab(T(-xy^T+B)T^{-1}) = \stab(-11^T+A).$
Since $A$ has zero diagonal, so do $A^z$ and $B$.  Note that $B$ can be obtained explicitly,
using the expression for $A$ in the proof of Lemma~\ref{lemma:13}. 
\end{proof}

\subsection{Appendix C: Remarks on the ratio $-\frac{\CM(A)}{\det(A)}$}

\noindent{\bf Remark 1.} If $A$ is an $n \times n$ nondegenerate square distance matrix for $n > 1$, then the ratio  $-\dfrac{\CM(A)}{\det(A)}$ has a very nice geometric interpretation:
$$-\dfrac{\CM(A)}{\det(A)} = \left|\dfrac{\CM(A)}{\det(A)}\right| = \dfrac{1}{2\rho^2},$$
where $\rho$ is the radius of the unique sphere circumscribed on the points used to generate $A$. 
This is proven in \cite[Proposition 9.7.3.7]{Berger}, where it is also shown that $\det(A) \neq 0$ not
only if $A$ is a  nondegenerate square distance matrix, but also if $A$ is a {\it degenerate} square
distance matrix corresponding to $n$ points in general position in $\RR^{n-2}$.  Since $\CM(A)$
vanishes in this case, we see that the ratio $-\dfrac{\CM(A)}{\det(A)}$ goes to zero
as $n$ points that are initially in general position in $\RR^{n-1}$ approach general position on a hyperplane of
dimension $n-2$.  
\medskip

\noindent{\bf Remark 2.}
The previous remark has important implications for the apparent ``fine-tuning'' that is involved
in eliminating spurious cliques by arranging points to be collinear, or coplanar, so that the corresponding
principal submatrix  $A_\sigma$ is degenerate (as in Figure 2B).  Since $-11^T + \varepsilon A_\sigma$
is only stable for 
$$0 < \varepsilon < -\dfrac{\CM(A_\sigma)}{\det(A_\sigma)} = \dfrac{1}{2\rho^2},$$ 
where $\rho$ is the radius of the circumscribed sphere, then by making the points $\{p_i\}_{i \in \sigma}$ corresponding to $A_\sigma$
{\it approximately} degenerate, $\rho$ can be made large enough so that $-11^T + \varepsilon A_\sigma$
is unstable -- without the fine-tuning required to make $A_\sigma$ exactly degenerate.

Similarly, exact solutions for $k$-skeleta of clique complexes, obtained using a $k$-sparse universal $S$
which is a {\it degenerate} square distance matrix, are also not as fine-tuned as they might first appear.
If, in fact, $S$ is a nondegenerate square distance matrix, corresponding to a configuration of $n$ points in $\RR^{n-1}$ 
that {\it approximately} lies on a $k$-dimensional plane, the value of
$\delta(S_\sigma)$ will be very small for any pattern of size $|\sigma| > k+1$; one can thus choose $\varepsilon$ large enough to ensure that 
$\geom_\varepsilon(S) = \{\sigma \subset [n] \mid |\sigma| \leq k+1 \},$ as in the case where $S$ is truly degenerate. 
\medskip

\noindent{\bf Remark 3.} It is quite simple to understand the scaling properties of $-\CM(A)/\det(A)$.  If $A$ is any $n \times n$ matrix, then  
$\CM(tA) = t^{n-1}\CM(A),$ while $\det(tA) = t^n \det(A)$, so
$$-\dfrac{\CM(tA)}{\det(tA)} = \dfrac{1}{t}\left(-\dfrac{\CM(A)}{\det(A)}\right),$$
independent of $n$.  If $A_{ij} = \norm{p_i-p_j}^2$, for $p_1,...,p_n \in \RR^{n-1}$, and we scale the 
position vectors so that $p_i \mapsto t p_i$ for each $i \in [n]$,
then $A \mapsto t^2 A$ and we have
$$-\dfrac{\CM(A)}{\det(A)} \mapsto \dfrac{1}{t^2}\left(-\dfrac{\CM(A)}{\det(A)}\right).$$
This is consistent with the fact that the radius $\rho$ of the circumscribed sphere scales
as $\rho \mapsto t \rho$ in this case (see Remark 1).
\medskip

\noindent{\bf Remark 4.}  Consider an $n \times n$ matrix $A$ satisfying the Hebbian conditions
$A_{ij} = A_{ji} \geq 0$ and $A_{ii} = 0$.  If $n$ is large, it is computationally intensive to test
whether or not $A$ is a nondegenerate square distance matrix using the criteria
of Lemma~\ref{lemma:realizability}, which potentially require
computing $\CM(A_\sigma)$ for all $\sigma \subset [n]$.  

On the other hand, our results imply that in order
to test whether or not a Hebbian matrix $A$ is a nondegenerate square distance matrix it is enough to check the stability of the matrix
$$-11^T + \varepsilon A \quad \mbox{for} \quad \varepsilon = \dfrac{1}{2}\left|\dfrac{\CM(A)}{\det(A)}\right|.$$
Here the factor of $1/2$ was chosen arbitrarily, and can be replaced with any number $0 < c < 1$.
For large $n$, this is a computationally efficient strategy, as it requires checking the eigenvalues of just one matrix.

\subsection{Appendix D: The input-output relationship of the network}

In this appendix we discuss the relationship between the inputs and outputs of the network with dynamics given by equation~\eqref{eq:dynamics}:
\begin{equation*}
 \dot x = -Dx+\left[W x + b  \right]_+\!,
\end{equation*}
with notation as described in Section~\ref{sec:background}.  While the inputs correspond to arbitrary vectors $b \in \RR^n$, the outputs of the network correspond to stable fixed points of the dynamics.  We consider two types of outputs: firing rate vectors $x^* \in \RR^n_{\geq 0}$ and binary patterns $\sigma = \supp(x^*),$ corresponding to subsets of co-active neurons at the fixed points.

The observations in this section all stem from a prior result, Proposition~\ref{prop:prop2.1} below~\cite{CDI12}.  Here we also use the notation ``$x < y$'' for vectors $x,y \in \RR^n$ to indicate that $x_i < y_i$ for each $i \in [n]$.  The symbols $>$ and $\leq$ are interpreted analogously.

 \begin{proposition}[{\cite[Proposition 2.1]{CDI12}}]\label{prop:prop2.1}
Consider the threshold-linear network $W$ (not necessarily symmetric) with dynamics given by equation~\eqref{eq:dynamics}, in the presence of a particular fixed input $b$.  Let $\sigma \subset [n]$ be a subset of neurons, and $\bar \sigma$ its complement.  Then, 
a point $x^* \in \RR^n$ with $x^*_\sigma >0$ and $x^*_{\bar \sigma} = 0$ is a fixed point if and only if 
\begin{itemize}
\item[(i)] $b_\sigma = (D - W)_\sigma x^*_\sigma$, and
\item[(ii)] $b_{\bar\sigma} \leq -W_{\bar\sigma \sigma}x^*_\sigma$,
\end{itemize}
where $W_{\bar\sigma \sigma}$ is the submatrix with rows and columns restricted to $\bar\sigma$ and $\sigma$, respectively.
In particular, if $\det(D-W)_\sigma \neq 0$, then there exists at most one nonnegative fixed point with support $\sigma$ and it is given by 
$$x_\sigma^* = (D-W)_\sigma^{-1}b_\sigma \quad \text{and} \quad x_{\bar\sigma}^* = 0,$$ 
provided that $(D-W)_\sigma^{-1}b_\sigma>0$ and properties (i) and (ii) hold.  Moreover, if $x^*$ is a fixed point and $b_{\bar\sigma} <  -W_{\bar\sigma \sigma}x^*_\sigma$, then $x^*$ is asymptotically stable if and only if $(-D+W)_\sigma$ is stable.
 \end{proposition}
 
 A simple corollary of this proposition is that for a given permitted set $\sigma \in \P(W)$, 
 the neurons in $\sigma$ may, depending on the input, be co-activated via any firing rate vector $x^* \in \RR^n_{\geq 0}$ that has support $\sigma$, although this vector is unique for a given input $b$.  
  
 \begin{corollary}
 Let $\sigma \in \P(W)$ be a permitted set of a threshold-linear network (not necessarily symmetric) with dynamics given by equation~\eqref{eq:dynamics}.  Then, for any $x^* \in \RR^n_{\geq 0}$ with $\supp(x^*) = \sigma$ (i.e., $x^*_\sigma > 0$ and $x^*_{\bar\sigma} = 0$), there exists an input $b \in \RR^n$ such that $x^*$ is the unique stable fixed point of~\eqref{eq:dynamics} whose subset of active neurons is exactly $\sigma$.  
 \end{corollary}
 
 \begin{proof}
 Choose any $b \in \RR^n$ such that $b_\sigma = (D-W)_\sigma x^*_\sigma$ and $b_{\bar\sigma} < -W_{\bar\sigma \sigma}x^*_\sigma$.  Observe that $(D-W)_\sigma$ is invertible because it is a stable matrix, since $\sigma \in \P(W)$.  Then, by Proposition~\ref{prop:prop2.1}, $x^*$ is the unique fixed point with support $\sigma$ in the presence of input $b$, and $x^*$ is asymptotically stable.  (Note, however, that stable fixed points with other supports may also arise for the same input $b$.)
 \end{proof}

The above results made no special assumptions about $W$; in particular, they  did not assume symmetricity.
Suppose now that $-W$ is nonnegative, as in the typical output of the Encoding Rule, and let $\sigma \in \P(W)$ be a permitted set.  Then $\sigma$ can be activated as an output binary pattern of the network by choosing any input $b \in \RR^n$ such that $b_\sigma = (D-W)_\sigma y$, for some $y \in \RR_{> 0}^{|\sigma|},$ and $b_{\bar\sigma}<0$.

The flexibility of possible output firing rate vectors, in contrast to the sharp constraints on output binary patterns of co-active neurons, suggests that the input-output relationship of threshold-linear networks should be regarded as fundamentally combinatorial in nature.

\subsection{Appendix E: Details related to generation of PF codes for Figure 4}

\hspace{.2in}
\vspace{-.1in}

\begin{wrapfigure}{r}{.45\linewidth}
\vspace{-.25in}
\begin{center}
\includegraphics[width=2in]{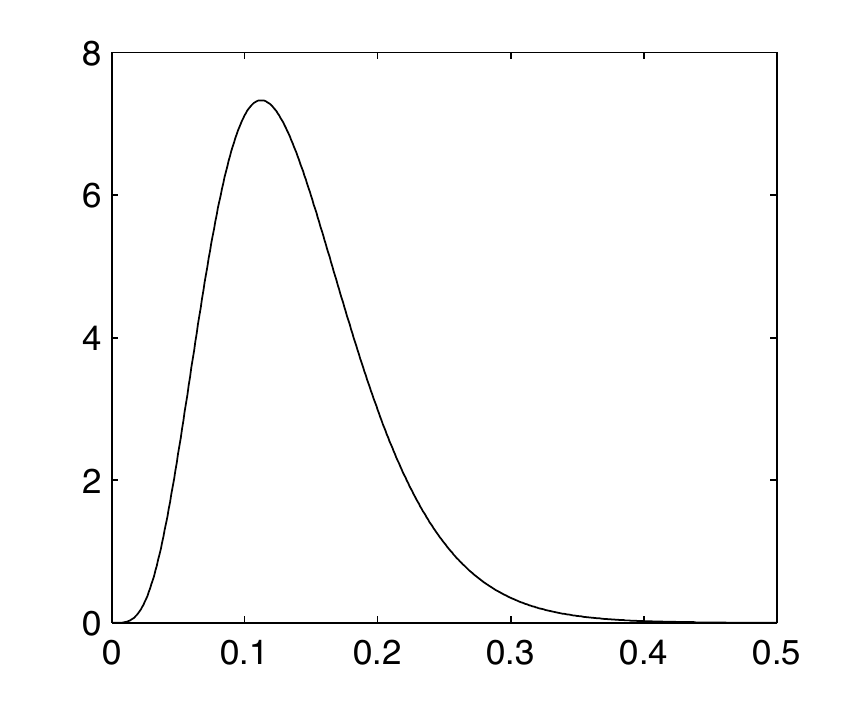}
\end{center}
\caption{\small Gamma distribution used for generating random place field radii; this fits the experimentally-observed mean and variability (see \cite[Figure 4B]{gap}).}
\end{wrapfigure}
%\vspace{-.3in}

\noindent To produce Figure 4, we generated random $k$-sparse PF codes with circular place fields, $n =$ 80-100 neurons, and $k = .1n$.  
For each code, $n$ place field centers were selected uniformly at random from a square box environment of side length $1$, and $n$ place field radii were drawn independently from an experimentally observed gamma distribution (Figure 5).  
We then computed the $2$-skeleton for each PF code, with pairwise and triple overlaps of place fields determined from simple geometric 
considerations.  The full PF code was obtained as the Helly completion of the $2$-skeleton (see Lemma~\ref{lemma:helly}).  Finally, to obtain the $k$-sparse
PF code, we restricted the full code to its $(k-1)$-skeleton, thereby eliminating patterns of size larger than $k$.  

\pagebreak

\bibliographystyle{spmpsci}
\bibliography{references}

\begin{thebibliography}{10}
\providecommand{\url}[1]{{#1}}
\providecommand{\urlprefix}{URL }
\expandafter\ifx\csname urlstyle\endcsname\relax
  \providecommand{\doi}[1]{DOI~\discretionary{}{}{}#1}\else
  \providecommand{\doi}{DOI~\discretionary{}{}{}\begingroup
  \urlstyle{rm}\Url}\fi

\bibitem{Amari1977}
Amari, S.: Dynamics of pattern formation in lateral-inhibition type neural
  fields.
\newblock Biol Cybern \textbf{27}(2), 77--87 (1977)

\bibitem{Amit1989}
Amit, D.: Modeling Brain Function: the World of Attractor Neural Networks.
\newblock Cambridge University press. (1989)

\bibitem{amit_spin-glass_1985}
Amit, D., Gutfreund, H., Sompolinsky, H.: Spin-glass models of neural networks.
\newblock Physical Review. A \textbf{32}(2), 1007{\textendash}1018 (1985)

\bibitem{AGS87}
Amit, D., Gutfreund, H., Sompolinsky, H.: Statistical mechanics of neural
  networks near saturation.
\newblock Annals of Physics NY 173, 30  (1987)

\bibitem{Amitbook}
Amit, D.J.: Modeling brain function.
\newblock Cambridge University Press, Cambridge (1989).
\newblock The world of attractor neural networks

\bibitem{HP}
Andersen, P., Morris, R., Amaral, D., Bliss, T., O'Keefe, J.: The Hippocampus
  Book.
\newblock Oxford University Press (2006)

\bibitem{Barth:2012}
Barth, A.L., Poulet, J.F.: Experimental evidence for sparse firing in the
  neocortex.
\newblock Trends Neurosci \textbf{35}(6), 345--355 (2012)

\bibitem{Barvinok2002}
Barvinok, A.: A course in convexity, \emph{Graduate Studies in Mathematics},
  vol.~54 (2002)

\bibitem{Ben-Yishai1995}
Ben-Yishai, R., Bar-Or, R.L., Sompolinsky, H.: Theory of orientation tuning in
  visual cortex.
\newblock Proc Natl Acad Sci U S A \textbf{92}(9), 3844--8 (1995)

\bibitem{Berger}
Berger, M.: Geometry. {I}.
\newblock Universitext. Springer-Verlag, Berlin (1994).
\newblock Translated from the 1977 French original by M. Cole and S. Levy,
  Corrected reprint of the 1987 translation

\bibitem{fiser:2011}
Berkes, P., Orban, G., Lengyel, M., Fiser, J.: Spontaneous cortical activity
  reveals hallmarks of an optimal internal model of the environment.
\newblock Science \textbf{331}, 83--87 (2011)

\bibitem{Blumenthal}
Blumenthal, L.M.: Theory and applications of distance geometry.
\newblock Oxford, at the Clarendon Press (1953)

\bibitem{BottTu}
Bott, R., Tu, L.W.: Differential forms in algebraic topology.
\newblock Springer-Verlag, New York (1982)

\bibitem{Bressloff12}
Bressloff, P.C.: Spatiotemporal dynamics of continuum neural fields.
\newblock J. Phys. A \textbf{45}(3) (2012)

\bibitem{cohen-grossberg:1983}
Cohen, M., Grossberg, S.: Absolute stability of global pattern formation and
  parallel memory storage by competitive neural networks.
\newblock IEEE Transactions on Systems, Man, and Cybernetics \textbf{SMC-13},
  815--826 (1983)

\bibitem{CDI12}
Curto, C., Degeratu, A., Itskov, V.: Flexible memory networks.
\newblock Bulletin of Mathematical Biology \textbf{74}(3), 590--614 (2012)

\bibitem{gap}
Curto, C., Itskov, V.: Cell groups reveal structure of stimulus space.
\newblock PLoS Comput Biol \textbf{4}(10) (2008)

\bibitem{neuro-coding}
Curto, C., Itskov, V., Morrison, K., Roth, Z., Walker, J.: Combinatorial neural
  codes from a mathematical coding theory perspective.
\newblock Neural Computation  (In press)

\bibitem{neural-ring}
Curto, C., Itskov, V., Veliz-Cuba, A., Youngs, N.: The neural ring: an
  algebraic tool for analyzing the intrinsic structure of neural codes.
\newblock arXiv:1212.4201 [q-bio.NC]  (2012)

\bibitem{DayanAbbott}
Dayan, P., Abbott, L.F.: Theoretical neuroscience.
\newblock MIT Press, Cambridge, MA (2001)

\bibitem{MoniqueLaurent}
Deza, M.M., Laurent, M.: Geometry of cuts and metrics, \emph{Algorithms and
  Combinatorics}, vol.~15.
\newblock Springer-Verlag, Berlin (1997)

\bibitem{ErmentroutTerman}
Ermentrout, G., Terman, D.: Mathematical Foundations of Neuroscience.
\newblock Springer (2010)

\bibitem{Seung:2003}
Hahnloser, R.H., Seung, H.S., Slotine, J.J.: Permitted and forbidden sets in
  symmetric threshold-linear networks.
\newblock Neural Comput \textbf{15}(3), 621--638 (2003)

\bibitem{HertzKroghPalmer}
Hertz, J., Krogh, A., Palmer, R.G.: Introduction to the theory of neural
  computation.
\newblock Addison-Wesley, Redwood City, CA (1991)

\bibitem{hillar:2012}
Hillar, C., Tran, N., Koepsell, K.: Robust exponential binary pattern storage
  in little-hopfield networks.
\newblock arXiv:1206.2081 [q-bio.NC]  (2012)

\bibitem{Hopfield:1982}
Hopfield, J.J.: Neural networks and physical systems with emergent collective
  computational abilities.
\newblock Proc. Natl. Acad. Sci. \textbf{79}(8), 2554--2558 (1982)

\bibitem{Hopfield:1984}
Hopfield, J.J.: Neurons with graded response have collective computational
  properties like those of two-sate neurons.
\newblock Proc. Natl. Acad. Sci. \textbf{81}, 3088--3092 (1984)

\bibitem{HornJohnson}
Horn, R., Johnson, C.: Matrix Analysis.
\newblock Cambridge University Press (1985)

\bibitem{Hromadka:2008}
Hrom{\'a}dka, T., Deweese, M.R., Zador, A.M.: Sparse representation of sounds
  in the unanesthetized auditory cortex.
\newblock PLoS Biol \textbf{6}(1) (2008)

\bibitem{HuffmanPless03}
Huffman, W., Pless, V.: Fundamentals of Error-Correcting Codes.
\newblock Cambridge University Press, Cambridge, MA (2003)

\bibitem{IHT11}
Itskov, V., Hansel, D., Tsodyks, M.: Short-term facilitation may stabilize
  parametric working memory trace.
\newblock Frontiers in Computational Neuroscience \textbf{5}, 1--19 (2011)

\bibitem{kenet:2003}
Kenet, T., Bibitchkov, D., Tsodyks, M., Grinvald, A., Arieli, A.: Spontaneously
  emerging cortical representations of visual attributes.
\newblock Nature \textbf{425}, 954--956 (2003)

\bibitem{luczak:2009}
Luczak, A., Bartho, P., Harris, K.: Spontaneous events outline the realm of
  possible sensory responses in neocortical populations.
\newblock Neuron \textbf{62}, 413--425 (2009)

\bibitem{MacWilliamsSloane83}
MacWilliams, F.J., Sloane, N.J.A.: The Theory of Error-Correcting Codes.
\newblock North-Holland, Amsterdam (1983)

\bibitem{PathIntegration}
McNaughton, B.L., Battaglia, F.P., Jensen, O., Moser, E.I., Moser, M.B.: Path
  integration and the neural basis of the 'cognitive map'.
\newblock Nat Rev Neurosci \textbf{7}(8), 663--78 (2006)

\bibitem{Muller:1996}
Muller, R.: A quarter of a century of place cells.
\newblock Neuron \textbf{17}(5), 813--822 (1996)

\bibitem{Cliquer}
Niskanen, S., Ostergard, P.: Cliquer - routines for clique searching (2010).
\newblock Available at http://users.tkk.fi/pat/cliquer.html

\bibitem{O:1976}
O'Keefe, J.: Place units in the hippocampus of the freely moving rat.
\newblock Exp. Neurol. \textbf{51}, 78--109 (1976)

\bibitem{ON:1978}
O'Keefe, J., Nadel, L.: The Hippocampus as a Cognitive Map.
\newblock Clarendon Press, Oxford, UK (1978)

\bibitem{OsborneBialek08}
Osborne, L., Palmer, S., Lisberger, S., Bialek, W.: The neural basis for
  combinatorial coding in a cortical population response.
\newblock Journal of Neuroscience \textbf{50}(28), 13,522--13,531 (2008)

\bibitem{Oxley2011}
Oxley, J.: Matroid theory, \emph{Oxford Graduate Texts in Mathematics},
  vol.~21, second edn.
\newblock Oxford University Press, Oxford (2011)

\bibitem{Petersen:1998:PNAS}
Petersen, C.C.H., Malenka, R.C., Nicoll, R.A., Hopfield, J.J.: All-or-none
  potentiation at {CA3}-{CA1} synapses.
\newblock PNAS \textbf{95}, 4732--4737 (1998)

\bibitem{RoudiTreves2003}
Roudi, Y., Treves, A.: Disappearance of spurious states in analog associative
  memories.
\newblock Phys. Rev. E \textbf{67}, 041,906 (2003)

\bibitem{Schoenberg}
Schoenberg, I.J.: Remarks to {M}aurice {F}r\'echet's article ``{S}ur la
  d\'efinition axiomatique d'une classe d'espace distanci\'es vectoriellement
  applicable sur l'espace de {H}ilbert''.
\newblock Ann. of Math. (2) \textbf{36}(3), 724--732 (1935).
\newblock \doi{10.2307/1968654}.
\newblock \urlprefix\url{http://dx.doi.org/10.2307/1968654}

\bibitem{Shriki:2003}
Shriki, O., Hansel, D., Sompolinsky, H.: Rate models for conductance-based
  cortical neuronal networks.
\newblock Neural Comput \textbf{15}(8), 1809--1841 (2003)

\bibitem{Seung:2002}
Xie, X., Hahnloser, R.H., Seung, H.S.: Selectively grouping neurons in
  recurrent networks of lateral inhibition.
\newblock Neural Comput \textbf{14}, 2627--46 (2002)

\bibitem{Ulanovsky:2013}
Yartsev, M.M., Ulanovsky, N.: Representation of three-dimensional space in the
  hippocampus of flying bats.
\newblock Science \textbf{340}(6130), 367--372 (2013)

\bibitem{yuste:2005}
Yuste, R., MacLean, J., Smith, J., Lansner, A.: The cortex as a central pattern
  generator.
\newblock Nat Rev Neurosci \textbf{6}, 477--83 (2005)

\end{thebibliography}

\end{document}